\documentclass[12pt]{amsart}

\usepackage{amssymb}
\usepackage{color}

\usepackage{enumitem}

\usepackage{bm}
\usepackage[T1]{fontenc}
\usepackage{fullpage}
\usepackage{mathtools}
\usepackage{multicol}

\usepackage[color=green!40,colorinlistoftodos,prependcaption]{todonotes}
\usepackage[all]{xy}

\makeatletter
\DeclareRobustCommand{\em}{%
  \@nomath\em \if b\expandafter\@car\f@series\@nil
  \normalfont \else \bfseries\itshape \fi}
\makeatother

\newtheorem{theorem}{Theorem}[section]
\newtheorem{corollary}[theorem]{Corollary}
\newtheorem{lemma}[theorem]{Lemma}
\newtheorem{proposition}[theorem]{Proposition}

\theoremstyle{definition}
\newtheorem{definition}[theorem]{Definition}

\newtheorem{example}[theorem]{Example}
\theoremstyle{remark}
\newtheorem{remark}[theorem]{Remark}

\numberwithin{equation}{section}
\numberwithin{figure}{section}

\newcommand{\B}[1]{{\mathbf #1}}
\newcommand{\C}[1]{{\mathcal #1}}

\newcommand{\OP}{\operatorname}

\def\L{\C L}
\def\oDelta{\Delta^\circ}
\def\oM{M^\circ}
\def\RR{\mathbb{R}}
\def\supp{\OP{supp}}
\newcommand{\xto}[1]{\xrightarrow{#1}}

\begin{document}

\title{Minimum Cost Super-Hedging in a Discrete Time Incomplete Multi-Asset Binomial Market}
\author{Jarek K\k{e}dra}
\address{University of Aberdeen and University of Szczecin}
\email{kedra@abdn.ac.uk}
\author{Assaf Libman}
\address{University of Aberdeen}
\email{a.libman@abdn.ac.uk}
\author{Victoria Steblovskaya}
\address{Bentley University}
\email{vsteblovskay@bentley.edu}


\begin{abstract}
We consider a multi-asset incomplete model of the financial market, where each of $m\geq 2$ risky assets follows the binomial dynamics, and no assumptions are made on the joint distribution of the risky asset price processes. We provide explicit formulas for the minimum cost super-hedging strategies for a wide class of European type multi-asset contingent claims. This class includes European basket call and put options, among others. Since a super-hedge is a non-self-financing arbitrage strategy, it produces non-negative local residuals, for which we also give explicit formulas. This paper completes the foundation started in previous work of the authors for the extension of our results to a more realistic market model.
\end{abstract}

\maketitle

\section{Introduction and the main result}\label{S:intro}

\subsection*{Overview}
The present paper is concerned with a discrete-time market model with one
risk-free asset $S_0$ and $m$ risky assets (e.g. stocks) $S_1,\ldots, S_m$.  Each risky asset price $S_i$ follows an $n$-step binomial model with parameters $0<D_i<R<U_i$, for $i=1,\ldots,m$. Here $R=1+r$ and $r$ is a periodic risk-free interest rate. More precisely, the asset price processes for the time steps $k=0,1,\ldots,n-1$
are given by
\begin{eqnarray*}\label{D:def_stock_price}
&& S_0(k+1)=S_0(k)\cdot R\quad \text{and}\quad
\\
&& S_i(k+1)=S_i(k)\cdot \psi_i,
\end{eqnarray*}
where the stock price ratios $\psi_i\in \{D_i,U_i\}$, $i=1,\ldots,m$ and the initial values $S_i(0)$ are given for all assets. No assumptions are made on the joint distribution of the stock price processes.

In this market, we consider a European type multi-asset contingent claim $F$ with maturity time $n$. 
It is well-known (see, e.g., \cite[Section 1.5.]{pliska}) that such a market model is incomplete for $m\ge 2$. There is no unique no-arbitrage price for $F$ at any time $k=0,\dots,n-1$, but rather an open interval of no-arbitrage prices $(C_{\min}(F,k),C_{\max}(F,k))$ at time $k=0,\ldots,n-1$. Each point in this interval is obtained as a discounted conditional expectation of $F$ with respect to a particular martingale measure. 
The boundaries of the interval are obtained by means of special extremal martingale measures. 

The boundaries $C_{\min}(F,k)$ and $C_{\max}(F,k)$ as well as the corresponding extremal hedging strategies are of great interest, both theoretically and for practical applications. However, they are hard to compute for a general contingent claim. More precisely, they may be efficiently computed by standard Linear Programming techniques for a single-step model with a small number of assets. However, if the number of time steps grows then applying such algorithms become computationally infeasible. See \cite[Section 4]{MR4553406} for more details on computations of the boundaries $C_{\min}(F,k)$ and $C_{\max}(F,k)$.

Practitioners are interested in closed form formulas that are easy to implement. Such formulas were obtained in our previous paper \cite{MR4553406} for the boundaries of a no-arbitrage price interval for a large class of European contingent claims, both path-dependent and path-independent. This class includes 
European basket call and put options, as well as Asian arithmetic average options. We review some of these results in Section \ref{SS:c-max}.

The boundaries $C_{\min}(F,0)$ and $C_{\max}(F,0)$ give rise to extremal hedging strategies. In particular, one can build a minimum cost super-hedge with the starting capital $C_{\max}(F,0)$. In this paper, we provide explicit formulas for such a super-hedging strategy at time $k=0,\ldots,n-1$. Notice that such hedging strategy is a non-self-financing arbitrage strategy. At each time step, the strategy produces a non-negative local residual. We provide explicit formulas for these local residuals as well. We consider the same large class of European contingent claims as in \cite{MR4553406}, but restrict our studies to the path-independent case only.

\subsection*{Context and existing results}
We see a motivation for our studies of this market model in potential extension of our results to the case of a more realistic multi-asset market model where the stock price ratio $\psi_i$, is no longer binomial, but rather is distributed over a bounded interval $[D_i,U_i]$, $i=1,\ldots,m$. In such extended binomial model, the explicit formulas derived in this paper may be used to build a proxy non-self-financing hedging strategy that uses a market price of a contingent claim as a starting set-up capital. See e.g. \cite{JKS-risk,MR3612260} and references therein for a similar approach in the case of a single-asset extended binomial model.

Before we formulate our main result, we briefly describe the existing relevant research of the multi-asset binomial model. A market model of this type was considered by many authors from various perspectives. In \cite{MR1611839}, a geometric approach was first applied and a general formula for the upper bound of the no-arbitrage option price interval was obtained in terms of an iteration of a certain transformation within the space of the Borel real valued measurable functions on $\RR^m$. However, explicit formulas were obtained only for a particular case of a two-asset single-step binomial market. The authors focused on self-financing  super-hedging strategies. They identified a minimum-cost self-financing super-hedge in general terms. 

Further, in \cite{MR2283328}, explicit formulas were obtained for the boundaries $C_{\min}(F,k)$ and $C_{\max}(F,k)$ in the two-asset multi-step binomial market, with no assumptions made on the joint distribution of the stock price ratios. The extremal non-self-financing hedging strategies were built explicitly for the same two-asset case. The results are extended to the case where the stock price ratios are distributed over a closed rectangle and the pay-off function is convex.

\subsection*{This work}
Let us now introduce some necessary definitions and formulate our main results.
Throughout we fix an $n$-step discrete time market model with risk-free asset $S_0$ and risky assets $S_1,\dots,S_m$ following the binomial model with price ratios $\psi_i\in \{D_i,U_i\}$, $0<D_i<R<U_i$, $i=1,\ldots,m$, as we described above.
\begin{definition}\label{D:class}
	A contingent claim $F$ is said to be of {\em class $\mathfrak{S}$} if 
	\begin{equation}\label{Eq:payoff}
		F = h\left( \sum_{i=0}^m \gamma_i S_i(n)\right),
	\end{equation}
	where $h\colon \RR\to \RR$ is a convex function, and 
	$\gamma_1,\dots,\gamma_m \geq 0$ (there is no restriction on $\gamma_0$).
\end{definition} 
The class $\mathfrak{S}$ includes European basket call and put options among other widely used contingent claims.
\begin{example}
In the case of a European basket call option with the strike price $K$, 
	 $h(x) = \max\,\{\,x\, , \,  0 \,\}$. In the case of a European basket put option with the strike price $K$, 
	 $h(x) = \max\,\{\,-x\, , \,  0 \,\}$. In both cases, $x=\sum_{i=0}^m \gamma_i S_i(n)$,
	 where $\gamma_i\geq 0$ for $i=1,2,\ldots,m$, $\sum_{i=1}^m \gamma_i=1$, and $\gamma_0=-\frac{K}{S_0(0)R^n}$.
	 \end{example}

\begin{definition}\label{defn:hedging strategy}
	A {\em hedging strategy} for a contingent claim $F$ is a sequence of vector valued random variables 
	$$
	\alpha(k)=(\alpha_0(k),\alpha_1(k),\ldots,\alpha_m(k)),
	$$ for $0 \leq k \leq  n-1$, whose component $\alpha_i(k)$, represents the weight of the $i$-th asset in the hedging portfolio at time $k$  $(0 \leq k \leq  n-1)$. The portfolio is set up at time $k$ for the time interval $[k,k+1)$. The setup cost $V_{\alpha}(k)$ of the hedging portfolio  at time $k$ is defined as follows:
	\begin{equation}\label{D:def_setup_cost}
		V_{\alpha}(k)=\sum_{i=0}^m\alpha_i(k) S_i(k).
	\end{equation}

	A {\em minimum cost super-hedging strategy} for a contingent claim $F$ is such hedging strategy that the following two conditions hold: 
	\begin{itemize}
		\item {\em Super-hedging:} At time $k+1$, before the hedging portfolio is rebalanced, the hedging portfolio value is no less than the upper bound of the no-arbitrage contingent claim price interval. That is, for any time $0 \leq k \leq  n-1$:
		\begin{equation}
			\sum_{i=0}^m\alpha_i(k) S_i(k+1) \geq C_{\max}(F,k+1). 
			\label{E:hedge condition}
		\end{equation}
		\item {\em Minimum cost:} 
		The strategy $\alpha(k)$ has the minimum setup cost $V_{\alpha}(k)$ of the hedging portfolio at time $k$. That is, for any strategy 
		$$
		\beta(k)=(\beta_0(k),\beta_1(k),\ldots,\beta_m(k)),
		$$ 
		for $0 \leq k \leq  n-1$, which satisfies the super-hedging condition \eqref{E:hedge condition}, we have 
		$$
		V_{\alpha}(k) \leq V_{\beta}(k)
		$$
		for any time $0 \leq k \leq  n-1$, where $V_{\beta}(k)=\sum_{i=0}^m\beta_i(k) S_i(k)$
		is a setup cost of $\beta(k)$ at time $k$.
	\end{itemize}
\end{definition}

\begin{remark}
It follows from the Duality of Linear Programming \cite[Chapter 5]{MR1485773} (see also \cite[Appendix]{pliska} for the explicit application to hedging) that 
\begin{equation}
V_{\alpha}(k) = C_{\max}(F,k).
\label{Eq:setup-cost}
\end{equation}
We will prove it directly; see the proof of Theorem 
\ref{T:minimum cost hedgiging theorem} in Section \ref{S:proof}.
\end{remark}

In order to formulate our main results, we need to introduce (first informally, later on more precisely, see Section \ref{section:formal model}) the notion of ``the state of the world'' $\omega$ at time $k$.

\begin{definition}\label{defn:state of the world}{\bf (informal)}
The state of the world $\omega$ at time $k$ is  a $k$-tuple $\omega =(\lambda^1,\dots,\lambda^k)$ where each $\lambda^j$ is an $m$-tuple of $0$'s and $1$'s which represents the dynamics of the $m$ risky asset prices from time step $j-1$ to time step $j$. 
A ``$0$'' in the $i^\text{th}$ position of $\lambda^j$ corresponds to the $i^\text{th}$ stock price ratio $\psi_i(\lambda^j)=D_i$, and a ``$1$'' in the $i^\text{th}$ position of $\lambda^j$ corresponds to the stock price ratio $\psi_i(\lambda^j)=U_i$.
\end{definition}

Thus, the price of the $i^\text{th}$ stock at time $k$ at the given state of the world $\omega$ at time $k$ is
\begin{equation}\label{eqn:S_(k) formula with psi}
S_i(k)(\omega) = 
S_i(0) \cdot \psi_i(\lambda^1) \cdots \psi_i(\lambda^k),\,  1 \leq i \leq m 
\end{equation}
Let us define the ``risk-neutral probabilities''
\begin{equation}\label{E:def bi}
b_i=\frac{R-D_i}{U_i-D_i}, \qquad (i=1,\ldots,m).
\end{equation}
Notice that $0<b_i <1$ and that by reordering the risky assets $S_1,\dots,S_m$ if necessary, we may assume that
\begin{equation}\label{E:bi non increasing}
0<b_m \leq \dots \leq b_1<1.
\end{equation}
For every $0 \leq j \leq m$ set
\begin{equation}\label{E:def qj}
	q_j = \left\{
	\begin{array}{ll}
		1-b_1 & \text{if } j=0 \\
		b_j-b_{j+1} & \text{if } 1 \leq j \leq m-1 \\
		b_m       & \text{if } j=m.
	\end{array}\right.
\end{equation}
Notice that $0\le q_j <1$ and that $\sum_{j=0}^m q_j =1$.
For integers $0 \leq i,j, \leq m$ set
\begin{equation}\label{E:def chii}
\chi_0(j)=R \qquad \text{and} \qquad
\chi_i(j)=\left\{
\begin{array}{rr}
U_i & i \leq j \\
D_i & i>j,
\end{array}
i=1,\dots,m.
\right.
\end{equation}

It will be convenient to consider the set of all sequences of length $k$ of elements from $\{0,1,\ldots,m\}$,
\begin{equation}\label{eqn:Pkm}
\C P_k(m)=\{0,\dots,m\}^k = \{ (j_1,\dots,j_k) : 0 \leq j_i \leq m\}.
\end{equation}
For any $J \in \C P_k(m)$ we write $J=(j_1,\dots,j_k)$ and  denote
\begin{align}
\label{E:def qJ}
q_J &= \prod_{t=1}^k q_{j_t},
\\
\label{E:def chiJ}
\chi_i(J)&= \prod_{t=1}^k \chi_i(j_t) \qquad (0 \leq i \leq m).
\end{align}
    
Our first result gives an explicit formula for $C_{\max}(F,k)$ at the state of the world $\omega=(\lambda^1,\dots,\lambda^k)$ at time $k$.
It should be compared with \cite[Section 6]{MR4553406}.
We stress that the assumption \eqref{E:bi non increasing} is in force.

\begin{proposition}\label{P:Cmax formula}
Let $F$ be a contingent claim of class $\mathfrak{S}$ in the $n$-step market model, where $n\geq1$, given by \eqref{Eq:payoff}.
For any $0 \leq k \leq n-1$ and any state of the world $\omega=(\lambda^1,\dots,\lambda^k)$
\begin{equation}\label{E:Cmax formula}
C_{\max}(F,k)(\omega) \, =  \,
R^{k-n} \sum_{J \in \C P_{n-k}(m)} q_J\cdot h\left( \sum_{i=0}^m \gamma_i \chi_i(J) S_i(k)(\omega) \right).
\end{equation}
\end{proposition}

Define $m$-tuples $\rho_0,\dots,\rho_m$ 
\begin{equation}\label{E:def rhoj}
	\rho_j=(\underbrace{1,\dots,1}_{\text{$j$ times}},0,\dots,0).
\end{equation}
Thus, $\rho_j$ describes the dynamics of the risky asset prices where the prices of the first $j$ assets $S_1,\dots,S_j$ went up and the prices of the remaining $m-j$ assets $S_{j+1},\dots,S_m$ went down.

Given the state of the world $\omega=(\lambda^1,\dots,\lambda^k)$ at time $0 \leq k \leq n-1$ (see Definition \ref{defn:state of the world}) and given $0 \leq j \leq m$ we denote by $\omega\rho_j=(\lambda^1,\dots,\lambda^k,\rho_j)$ the state of the world at time $k+1$, where the risky asset price dynamics up to time $k$ is described by $\omega$, and their price dynamics from time $k$ to time $k+1$ is described by $\rho_j$.

The main result of this paper is the next theorem that gives an explicit formula for the minimum cost super-hedging strategy for a contingent claim of class $\mathfrak{S}$.

\begin{theorem}\label{T:minimum cost hedgiging theorem}
Let $F$ be a contingent claim of class $\mathfrak{S}$ in an $n$-step discrete time binomial market given by \eqref{Eq:payoff}.
Then there exists a minimum cost super-hedging strategy $\alpha$.
At any state of the world $\omega$ at time $0 \leq k<n$ the components of $\alpha$ are given by
\begin{equation}\label{E:solution}
\begin{pmatrix}
\alpha_0(k)(\omega)\\
\alpha_1(k)(\omega)\\
\vdots\\
\alpha_m(k)(\omega)\\
\end{pmatrix} 
=
T_k(\omega)^{-1}\cdot N^{-1}\cdot Q\cdot
\begin{pmatrix}
C_{\max}(F,k+1)(\omega\rho_0)\\
C_{\max}(F,k+1)(\omega\rho_1)\\
\vdots\\
C_{\max}(F,k+1)(\omega\rho_m)\\
\end{pmatrix},
\end{equation}
where $T_k(\omega),N,Q$ are the following $(m+1)\times(m+1)$ matrices 
\begin{eqnarray}\label{D:def_matrices}
Q&=&\left(
\begin{smallmatrix*}[r]
1  &  0 & 0 & \cdots \cdots & 0 & 0 \\
-1 &  1 & 0 & \cdots \cdots & 0 & 0 \\
0  & -1 & 1 & \cdots \cdots & 0 & 0 \\
\vdots & \vdots & \vdots & \cdots \cdots &  & \vdots \\
0 &   0 & 0 & \cdots \cdots & 1 & 0 \\
0 &   0 & 0 & \cdots \cdots & -1 & 1 \\
\end{smallmatrix*}
\right),
\\
\nonumber
N &=& 
\left(
\begin{smallmatrix*}[r]
1 & D_1       & D_2      & \cdots    & D_m \\
0 & U_1-D_1   & 0        & \dots     & 0                \\
0 & 0         & U_2-D_2  &           & 0                \\
\vdots &      &          & \ddots    & \vdots           \\
0 &  0        & 0        &           &    U_m-D_m  
\end{smallmatrix*}
\right),
\\
\nonumber
T_k(\omega)&=&
\left(
\begin{smallmatrix*}[r]
R S_0(k)(\omega)    &            &                &  \\
              & S_1(k)(\omega) &                &  \\
              &            &  \ddots        &  \\
              &            &         & S_m(k)(\omega)  
\end{smallmatrix*}
\right).
\end{eqnarray}
Moreover, at any time $k$ the setup cost of the hedging portfolio is equal exactly to $C_{\max}(F,k)$.
That is, for any state of the world $\omega$ at time $k$
\[
V_\alpha(k)(\omega) = C_{\max}(F,k)(\omega).
\]
\end{theorem}

\begin{remark}
The minimum cost super-hedging strategy $\alpha$ is non-anticipating: the components $\alpha_0(k),\dots,\alpha_m(k)$ of the hedging portfolio at time $k$ are computed by means of {\em only} the prices of the assets $S_0,\dots,S_m$ at time $k$, i.e $S_0(k),\dots,S_m(k)$.
Indeed, the matrix $T_k$ only depends on the ``current prices'' at time $k$. Moreover, the upper bound of the no-arbitrage contingent claim price interval $C_{\max}(F,k+1)(\omega\rho_j)$ at time $k+1$ is computed by means of the ``current prices'' of $S_i$ at time $k$. Notice that by \eqref{eqn:S_(k) formula with psi} $S_i(k+1)(\omega \rho_j) = S_i(k)(\omega) \cdot \chi_i(j)$, so Proposition \ref{P:Cmax formula} implies
\[
C_{\max}(F,k+1)(\omega\rho_j) = 
R^{k+1-n} \sum_{J \in \C P_{n-k-1}(m)} q_J\cdot h\left( \sum_{i=0}^m \gamma_i \chi_i(J)\chi_i(j) \cdot  S_i(k)(\omega) \right).
\]
\end{remark}


\subsection*{Organization of the paper}

In Section \ref{S:finite} we recall necessary definitions and facts related to finite probability spaces and adapt them to the context of this paper. We also introduce some model specific definitions and prove necessary propositions.   

In Section \ref{section:formal model} we introduce a formal model for our discrete time multi-asset Binomial model.

In Sections \ref{section:supermodular} and \ref{section:RN C_max} we provide an overview of results obtained in our previous paper \cite{MR4553406}. However, we adapt the terminology to the context of this paper and provide a more thorough proofs of main propositions. In Section \ref{section:supermodular}, we introduce important definitions: a supermodular random variable and a supermodular vertex (a special measure). Theorem \ref{theorem:supermodular vertex maximium} provides a crucial tool for pricing contingent claims in our model. In Section \ref{section:RN C_max} we describe the set of risk-neutral and martingale measures and prove Proposition \ref{P:Cmax formula} which provides an explicit formula for the upper bound of a no-arbitrage price interval for a contingent claim of class $\mathfrak{S}$ .

The proof of Theorem \ref{T:minimum cost hedgiging theorem} is presented in Section \ref{S:proof}. In Section \ref{S:Residuals}, we introduce local and accumulated residuals of the minimum cost super-hedging strategy and provide explicit formulas for them. In Section \ref{S:example}, we work out a simple example of a $2$-step model with two risky assets in which we evaluate all crucial formulas from the paper, first abstractly, then with concrete numbers.
We conclude the paper with Section \ref{S:conclusion}  where we summarize our findings and discuss a more realistic market model for which our present studies build a foundation.

\section{Finite probability spaces}\label{S:finite}


In this section we recall a number of basic facts about finite
probability spaces and provide necessary definitions which will be needed in the subsequent sections.

Let $(A,p)$ be a finite probability space. That is, $A$ is a finite
set and $p\colon 2^A\to [0,1]$ is a probability measure defined on
all subsets of $A$. The measure defines a function (denoted by 
abuse of notation by the same symbol) $p\colon A\to [0,1]$ 
such that
\begin{equation}\label{Eq:pdf}
\sum_{a\in A} p(a) =1.
\end{equation}
We will refer to both as probability measures\footnote{
There is no standard terminology here. For example, Billingsley calls the value $p(a)$ a mass
\cite[Example 2.9 and 10.2]{MR1324786} and Shreve \cite{MR2049045} uses the term probability measure for both the
measure and the function.}.
Conversely, any function as in \eqref{Eq:pdf} defines a probability measure by
$$
p(B) = \sum_{a\in B}p(a),
$$
where $B\subseteq A$ is an event. 
Let 
$$
\Delta(A) = \left\{p\colon A\to [0,1]\ |\ \sum_{a\in A}p(a)=1\right\}
$$
denote the set of all probability measures on $A$. 
Let $n\in \B N$ be the number of elements in $A$.
Notice that
$\Delta(A)\subseteq \RR^{n}$ is a standard simplex and, in particular,
it is a compact set.

We say that $p\in \Delta(A)$ is {\em non-degenerate} if $p(a)>0$ for all $a\in A$.
This is the case if and only if the only event of probability zero is
the empty set. The set of non-degenerate probability measures
is denoted by $\Delta^{\circ}(A)$.

Once $p \in \Delta(A)$ is fixed, a random variable on $A$ is merely a function
$X \colon A \to \RR$.  For this reason we will refer to any function $X \colon
A \to \RR$ as a {\em random variable} even when the probability measure $p$ is
not specified. The vector space of all functions $X\colon A\to \RR$ will be denoted
by $\RR ^A$. It is isomorphic to $\RR^n$, where $n$ is the number of elements
of the set $A$. Later on we will consider it equipped with the standard inner
product.
A function $X \colon A \to \RR^m$ is called a {\em random
vector} with components $X_i \colon A \to \RR$, $i=1,2,\ldots,m$.

Suppose $p \in \Delta(A)$ and $B \subseteq A$ and $X \colon A \to \RR^m$.
Define
\[
\int_B X \, d p = \sum_{b \in B} X(b) p(b)\in \RR^m.
\]
With this notation, the expectation of $X$ with respect to $p$ is:
\[
E_p(X) = \sum_{a \in A} X(a) p(a) = \int_A X \, dp.
\]
Observe that once a random vector $X$ is fixed the function
\[
E_{-}(X) \colon \Delta(A) \xto{ \ p \mapsto E_p(X) \ } \RR^m
\]
is the restriction of a linear function $\RR^A \to \RR^m$.
In particular, it is continuous.

\subsection{Product and conditional probabilities}\label{subsec:product and conditional prob}
Let $p \in \Delta^{\circ}(A)$ be a non-degenerate probability measure.
Recall that for any non-empty $B\subseteq A$ the conditional probability $p(-|B)\colon 2^A\to [0,1]$ 
is defined by
$$
p(C|B) = \frac{p(C\cap B)}{p(B)}.
$$
Notice that $p(-|B)$ can be considered as a probability measure on $B$ and it 
will be denoted by $p_{|B}\colon B\to [0,1]$. It is given by
$$
p_{|B}(b) = \frac{p(b)}{p(B)}.
$$
In particular, $p_{|B}\in \Delta(B)$.
In this setup the conditional expectation of a random vector $X \colon A \to \RR^m$
is given by
\[
E_p(X|B) = \int_A X \, dp_{|B} = \int_B X \, dp_{|B} = \tfrac{1}{p(B)} \int_B X\, dp.
\]

The sample spaces we deal with in the present paper are multiple Cartesian products
and we will most often need conditional expectations with respect to their factors.
For this reason we will introduce the following, somewhat arbitrary, notation.

Let $A',A''$ be finite sets.
If $B \subseteq A'$ we write
\begin{equation}\label{eqn:bar notation for events}
\overline{B} = B \times A''\subseteq A'\times A''.
\end{equation}
If $a' \in A'$ we write $\overline{a'}=\overline{\{a'\}}=\{a'\} \times A'' \approx A''$,
where $\approx$ denotes the bijection $(a',a'') \mapsto a''$.

Let $X'\colon A'\to \RR^m$ and $X''\colon A''\to \RR^m$ be random vectors. 
They define a random vector $X' \otimes X''\colon A' \times A''\to \RR^m$
with components given by
\[
(X' \otimes X'')_i \ \colon \ A' \times A'' \xto{\ \  (a',a'') \mapsto X'_i(a') \cdot X''_i(a'') \ \ } \RR, \qquad i=1,\dots,m. 
\]
If $X \colon A' \times A'' \to \RR^m$ is a random vector and $a' \in A'$, 
we obtain a random vector
\[
X(a',-) \colon A'' \xto{ \ \ a'' \mapsto X(a',a'') \ \ }\RR^m. 
\]

\begin{lemma}\label{lem:products conditional general}
\hfill
\begin{enumerate}[label=(\roman*)]
\item
\label{lem:products conditional general:product}
If $p' \in \Delta(A')$ and $p'' \in \Delta(A'')$ then $p' \otimes p'' \in \Delta(A' \times A'')$.
It is the product measure.

\item
\label{lem:products conditional general:conditional}
Let $a' \in A'$ be fixed.
The conditional probability yields a function
\[
\oDelta(A' \times A'') \xto{\ p \mapsto p_{|\overline{a'}} \ } \oDelta(A'').
\]

\item
\label{lem:products conditional general:tensor-conditional}
Let $p' \in \oDelta(A')$ and $p'' \in \oDelta(A'')$.
Set $p=p' \otimes p''$.
Then $p_{|a'} = p''$ for any $a' \in A'$.

\item
\label{lem:products conditional general:E in stages}
Suppose $p' \in \Delta(A')$ and $p'' \in \Delta(A'')$ and $X$ a random vector on $A' \times A''$.
Let $Y$ be the random vector on $A'$ defined by $Y(a')=E_{p''}(X(a',-))$.
Then
\[
E_{p'}(Y)=E_{p' \otimes p''}(X).
\]
\item \label{lem:products conditional general:box restriction}
For any $p \in \oDelta(A' \times A'')$ , any random vector $X$ on $A' \times A''$ and any $a' \in A'$
\[
E_{p}(X|\overline{a'}) = E_{p_{|\overline{a'}}}(X(a',-)).
\]
\end{enumerate}
\end{lemma}

\begin{proof}
Parts \ref{lem:products conditional general:product} and \ref{lem:products conditional general:conditional} 
follow directly from definitions.
\\

\noindent
\ref{lem:products conditional general:tensor-conditional}
For any $a'' \in A'' \approx \{a'\} \times A''$,
\[
p_{|\overline{a'}}(a'') =p(a',a'')/p(\{a'\} \times A'') = p'(a')p''(a'')/p'(a')p''(A'') = p''(a'').
\]
\ref{lem:products conditional general:E in stages}
We calculate
\[
E_{p'}(Y) = \sum_{a' \in A'} Y(a')p'(a') =
\sum_{a' \in A'} \sum_{a'' \in A''} X(a',a'') p''(a'') p'(a')
=E_{p' \otimes p''}(X).
\]
\ref{lem:products conditional general:box restriction}
Since $\overline{a'}=\{a'\} \times A''$, we use part \ref{lem:products conditional general:tensor-conditional} to compute
\[
E_p(X|\overline{a'}) =
\int_{\overline{a'}} X \, dp_{|\overline{a'}} =
\sum_{a'' \in A''} X(a',a'') p_{|\overline{a'}} =
E_{p_{|\overline{a'}}}(X(a',-)).
\]

\end{proof}

\subsection{Martingale measures}\label{subsection:martingale measures}
Recall that a sequence of random variables $X^0,X^1,\dots$ adapted to
a filtration $\{\C F_k\}$ of the corresponding sample space
is called a {\em discrete martingale} if
$$
E(X^{k+1}|\C F_k) = X^k
$$ 
(see, for example, \cite[Definition 9.1]{MR3059814}). 

In our situation we fix a positive integer $n\in \B N$ and a finite set
$A$. We consider the sample space $A^n$ together with the filtration 
$\C F_0\subseteq \C F_1\subseteq \dots \subseteq \C F_n=2^{A^n}$,
where the $\sigma$-algebra $\C F_k$ is generated by a partition $\C P_k$
of $A^n$ defined by
$$
\C P_k(a_1,\ldots,a_k) =\{(a_1,\ldots,a_k)\}\times A^{n-k} = \overline{(a_1,\ldots,a_k)}.
$$
Notice that the sets of this partition are in bijection with elements of $A^k$.
We call the above filtration $\{\C F_k\}$ 
the {\em product filtration}. Observe that an $\C F_k$-measurable
random variable $X$ defined on $A^n$ is constant on the sets of $\C P_k$. Equivalently,
$X(a_1,\ldots,a_n)$ depends on the first $k$ coordinates only or
that $X$ factors as
$$
X\colon A^n\to A^k\to \RR,
$$
where $A^n\to A^k$ is the projection onto the first $k$ factors.
It is thus convenient to consider $\C F_k$-measurable random variables
as functions on $A^k$.

For any random variable $X$ on $A^n$, the conditional expectation
$E(X|\C F_k)$ is an $\C F_k$-measurable random variable and it satisfies
the following identity:
\begin{equation}
p(\overline{a})E(X|\C F_k)(\overline{a}) = \int_{\overline{a}}Xdp,
\label{Eq:cond-exp}
\end{equation}
where $a\in A^k$ and $\overline{a}=\{a\}\times A^{n-k}$ (see, for example,
\cite[Section 10.1]{MR2977961}).

Consider an $\RR^m$-valued {\em $n$-step random process} $\{X^k\}_{k=0,\dots,n}$ on $A^n$ adapted to the product filtration.
That is, a sequence of $n+1$ random vectors
$$
X^0,X^1,\ldots,X^n\colon A^n\to \RR^m,
$$
such that $X^k$ if $\C F_k$-measurable. 

Fix some $c=(c_1,\ldots,c_m) \in \RR^m$ with $c_i>0$ and define the $n$-step random process $\{Y^k\}_{k=0,\dots,n}$ as follows:
\begin{equation}\label{eq:Yk}
	Y^k = (c_1^{-k}X_1^k,\ldots,c_m^{-k}X_m^k),
\end{equation}
where, as above, $X^k_i$ denotes the $i^\text{th}$ component of $X^k$.

\begin{definition}\label{def:martingale}
A {\em martingale measure for the $n$-step random process $\{X^k\}_{k=0,\dots,n}$ and $c\in \RR^m$} 
is a probability measure $p \in \Delta(A^n)$ such that the process  $\{Y^k\}_{k=0,\dots,n}$
defined in \eqref{eq:Yk} is a discrete martingale with respect to $p$. That is,
\begin{equation}
E_p(Y^{k+1}\ |\ \C F_k) = Y^k.
\label{Eq:martingale-c}
\end{equation}
\end{definition}
Observe that the above defining equation can be equivalently rewritten as follows,
where $a'\in A^k$, $i=1,\ldots, m$ and $0\leq k< n$.
\begin{align}
E_p(c_i^{-(k+1)}X_i^{k+1}\ |\ \C F_k) &= c_i^{-k}\cdot X_i^k \nonumber\\
E_p(X_i^{k+1}\ |\ \C F_k)             &= c_i \cdot X_i^k     \nonumber\\ 
p(\overline{a'})E_p(X_i^{k+1}\ |\ \C F_k)(\overline{a'}) &= p(\overline{a'})\cdot c_i \cdot X_i^k(\overline{a'})
\nonumber \\
\int_{\overline{a'}} X_i^{k+1} \, dp &= p(\overline{a'}) \cdot c_i \cdot X_i^{k}(a') 
\label{eqn:martingale ell equiv}
\end{align}
Since $X^k$ is constant on $\overline{a'}$ with value $X^k(a')$, the last equation is
equivalent to
\begin{equation}\label{eqn:martingale ell}
\int_{\overline{a'}} X_i^{k+1} \, dp = c_i \int_{\overline{a'}} X_i^{k} \, dp, \qquad i=1,\dots,m.
\end{equation}
In the subsequent computations we will mostly use equation \eqref{eqn:martingale ell equiv}.

The basic example in this paper is a risk-neutral measure (as defined by \eqref{eqn:risk neutral def})
which is a non-degenerate martingale measure for the stock price process $\{S(k)\}_{k=0,\dots,n}$ and
$c=(R,R,\ldots,R)$. Here $S(k)$ is a vector of stock prices at time $k$: $S(k)=\left(S_1(k),\dots,S_m(k)\right)$.

We denote the set of all martingale measures for $\{X^k\}_{k=0,\dots,n}$ and $c\in \RR^m$ by
\[
M(X^k;c) = \{ p \in \Delta(A^n) \ : \ \text{$p$ is a martingale measure for $\{X^k\}_{k=0,\dots,n}$ and $c$}\}
\]
and by $\oM(X^k;c)$ the set of non-degenerate martingale measures.
Notice that \eqref{eqn:martingale ell} is a system of linear equations for the unknown $p$.
It follows that $M(X^k;c)$ is the intersection of the simplex $\Delta(A^n)$ with a number of
hyperplanes. We thus deduce the following observation.

\begin{proposition}
The set of martingale measures $M(X^k;c)$ is a bounded polytope in $\RR^A$, 
hence a compact convex set in $\RR^A$.
\qed
\end{proposition}

\subsection{Random processes with uniform quotients}\label{subsection:unform random process}
In this subsection we fix a function $\phi =(\phi_1,\ldots,\phi_m) \colon A \to \RR^m$ such that $\phi_i(a)>0$ for 
all $i=1,\ldots, m$.

\begin{definition}
An $n$-step  $\RR^m$-valued random process $\{X^k\}_{k=0,\dots,n}$ is said to have {\em uniform quotient} $\phi$ if  $X^0>0$ and
\[
X^k_i(a_1,\dots,a_n) \, = \, X^0_i \cdot \phi_i(a_1) \cdots \phi_i(a_k).
\]
\end{definition}
The basic example we are interested in is the stock price process $\{S(k)\}_{k=0,\dots,n}$ following the
binomial model. In this case $\phi_i(a_j) \in \{D_i,U_i\}$ (see Section \ref{section:formal model}).
In what follows we give a general argument which, specified to the financial models
in the subsequent sections, says that martingale measures can be defined 
using either the stock price process or the stock price ratios.

\begin{definition}\label{def:M_n(phi;c)}
We introduce here the notation of $M_n(\phi,c)$, $M_n(A,c)$, $M_n^{\circ}(\phi,c)$,
$M_n^{\circ}(A,c)$ for various subsets of the set of martingale measures.

Consider a positive valued  $\phi \colon A \to \RR^m$ and  $n \geq 0$.
Given $1 \leq k \leq n$ set $\Phi^k \colon A^n \to \RR^m$ to be the composition
\[
\Phi^k \colon A^n \xto{\ \pi_k \ } A \xto{\ \phi \ } \RR^m,
\]
where $\pi_k\colon A^n\to A$ is the projection onto the $k$-th factor.
Given $c \in \RR^m$ let
\[
M_n(\phi;c)
\]
be the subset of $\Delta(A^n)$ containing measures
$p$ such that for any $0 \leq k \leq n-1$ and any $a' \in A^k$,
\[
\int_{\overline{a'}} \Phi_i^{k+1} \, dp = c_i p({\overline{a'}}).
\]
Written explicitly, this is the condition
\begin{equation}\label{eqn:uniform distillation}
\sum_{a \in A} \phi_i(a) p(\overline{a'a}) = c_i p(\overline{a'}).
\end{equation}
If $A \subseteq \RR^m$ and $\phi \colon A \to \RR^m$ is the inclusion then instead of  $M_n(\phi;c)$ we write
\[
M_n(A;c).
\]
The subsets of $M_n(\phi,c)$ and of $M_n(A,c)$ consisting of
non-degenerate measures are denoted $\oM_n(\phi;c)$ and $\oM_n(A;c)$, respectively.
\end{definition}

If $p$ is non-degenerate then \eqref{eqn:uniform distillation} 
is equivalent to the condition that for any $a' \in A^k$
\[
E_p(\Phi^k|\overline{a'}) = c.
\]
We also notice that if $n=1$ then $M_1(\phi;c)$ consists of such $p \in \Delta(A)$ that
\begin{equation}\label{eqn:M_q(A;c)}
E_p(\phi)=c.
\end{equation}

\begin{proposition}\label{prop:M_n of uniform}
Suppose that $\{X^k\}_{k=0,\dots,n}$ is a random process on $A$ with  uniform quotient $\phi>0$ and that $X^0>0$.
Then
\[
M(X^k;c) = M_n(\phi;c).
\]
\end{proposition}

\begin{proof}

An element of $M(X^k;c)$ is characterized by equation \eqref{eqn:martingale ell equiv}.
Dividing both sides of this equation by the positive number 
$X_i^k(a')=X^0_i \cdot \phi_i(a'_1) \cdots \phi_i(a'_k)$ 
yields $\int_{\overline{a'}}\Psi_i^{k_1}dp = c_ip(\overline{a'})$, which characterizes
an element of $M_n(\phi;c)$.
\end{proof}

\begin{lemma}\label{lem:tensor M psi c}
Suppose that $p' \in M_{n'}(\phi;c)$ and $p'' \in M_{n''}(\phi;c)$. 
Then $p' \otimes p'' \in M_{n'+n''}(\phi;c)$.
\end{lemma}

\begin{proof}
Set $p=p' \otimes p''$.
Consider $0 \leq k < n$ and $b \in A^k$.
Suppose $a \in A$.
If $k<n'$ then $p(\overline{ba})=p'(\overline{ba})$ and $p(\overline{b})=p'(\overline{b})$.
So \eqref{eqn:uniform distillation} holds for $p$ since it holds for $p' \in M_{n'}(\phi;c)$.
If $n' \leq k <n$ then we can write $b=a'a''$ for some $a' \in A^{n'}$ and $a'' \in A^{k-n'}$.
Then $p(\overline{ba})=p'(a')p''(\overline{a''a})$ and $p(\overline{b}) = p'(a') p''(\overline{a''})$, and \eqref{eqn:uniform distillation} holds for $p$ since it holds for $p'' \in M_{n''}(\phi;c)$.
\end{proof}

We remark that for any $a' \in A^k$ where $0 \leq k \leq n$ there is a natural bijection between
$\overline{a'} =\{a'\} \times A^{n-k}$ and  $A^{n-k}$.

\begin{lemma}\label{lem:conditional oM}
Suppose that $p \in \oM_n(\phi;c)$ and $b \in A^k$ for some $0 \leq k  \leq n$.
Then $p_{|\overline{b}} \in \oM_{n-k}(\phi;c)$.
\end{lemma}

\begin{proof}
Clearly $p_{|\overline{b}}>0$.
Suppose that $a' \in A^j$ for some $0 \leq j < n-k$ and that $a \in A$.
Then $p_{|\overline{b}}(\overline{a'a}) = \tfrac{1}{p(\overline{b})}p(\overline{ba'a})$ and
$p_{|\overline{b}}(\overline{a'}) = \tfrac{1}{p(\overline{b})} p(\overline{ba'})$.
Apply \eqref{eqn:uniform distillation} to $p \in \oM(\phi;c)$ and $ba' \in A^{k+j}$, and divide both sides by $p(\overline{b})$.
\end{proof}

\begin{corollary}\label{cor:oM(L,b) conditional surjective}
Consider some $0<k<n$ and suppose that $\oM_k(\phi;c)$ 
is not empty.
Let $a' \in A^k$.
Then
\[
\oM_n(\phi;c) \xto{\ p \mapsto p_{|\overline{a'}} \ } \oM_{n-k}(\phi;c)
\]
is surjective.
\end{corollary}

\begin{proof}
The assignment $p \mapsto p_{|\overline{a'}}$ is into $\oM_{n-k}(\phi;c)$ by Lemma \ref{lem:conditional oM}.
Consider some $p'' \in \oM_{n-k}(\phi;c)$.
Choose some $p' \in \oM_k(\phi;c)$ and set $p=p' \otimes p''$.
It is clear that $p>0$ and $p \in M_n(\phi;c)$ by Lemma \ref{lem:tensor M psi c}.
So $p \in \oM_n(\phi;c)$.
By Lemma \ref{lem:products conditional general}\ref{lem:products conditional general:tensor-conditional} $p_{|\overline{a'}}=p''$.
\end{proof}

\section{A formal model for discrete time multi-asset binomial market}\label{section:formal model}

In the paper \cite{MR4553406}, we introduced and developed in detail an
information structure of our market model. We also described in detail the set
of martingale and risk-neutral measures. Here we present a more formal and
abbreviated version of that description which is sufficient for the purposes of
this paper. We refer the reader to \cite{MR4553406} for more details.

\subsection{A single time step model}\label{subsection:single-step}
Recall that in our market model each risky asset $S_i$ follows the binomial model with parameters $(D_i,U_i)$, $i=1,\dots ,m$.
A natural sample space for a single time step market model (i.e $n=1$) is the set 
\[
\L=\{0,1\}^m
\]
of all the sequences of length $m$ consisting of $0$'s and $1$'s. 
We denote its elements by $\lambda=(\lambda(1),\dots,\lambda(m))$ and view $\L$ as a subset of $\RR^m$. 
Each $m$-tuple $\lambda$ represents ``the state of the world'' at time $n=1$. 
Namely, if $\lambda(i)=0$, then the  $i^\text{th}$ stock price dropped down at time $n=1$ with $S_i(1)=S_i(0)D_i$; If $\lambda(i)=1$, then the $i^\text{th}$ stock price went up at time $n=1$ with  $S_i(1)=S_i(0)U_i$. 
Thus, $S_i(0)$ and $S_i(1)$ are  random variables over $\L$ with $S_i(0)$ constant (initial prices of stocks) and $S_i(1)=S_i(0) \psi_i$ where $\psi_i$ are the random variables of the stock price ratios ($i=1,\dots,m)$) at time $n=1$:
\begin{equation}\label{E:def psii def}
\psi_i(\lambda) = \left\{
\begin{array}{ll}
D_i           & \text{if $\lambda(i)=0$} \\
U_i           & \text{if $\lambda(i)=1$}
\end{array}\right.
\end{equation}
Let the random variables 
\[
\ell_i \colon \L \to \{0,1\}, \qquad i=1,\ldots,m
\]
represent a single-step dynamics of the the $i^\text{th}$ stock, i.e
\begin{equation}\label{E:def elli}
\ell_i( \lambda) = \lambda(i).
\end{equation}
Thus, $\ell_i$ is the restriction to $\L$ of the linear projection  $\pi_i \colon \RR^m \to \RR$ onto the $i^\text{th}$ factor. 
With the above notation, we have for each $\lambda \in \L$
\[
\psi_i(\lambda)=D_i+(U_i-D_i)\lambda(i)
\]
or, equivalently, for every $i=1,\ldots,m$
\begin{equation}\label{E:def psii}
\psi_i = D_i+(U_i-D_i)\ell_i,
\end{equation}
Notice that $\psi_i$ is the restriction to $\L$ of the affine function $\tilde{\psi}_i \colon \RR^m \to \RR$ 
\[
\tilde{\psi}_i(x_1,\dots,x_m) = D_i+(U_i-D_i)x_i.
\]
Since the price ratio of $S_0$ is always $R$ it is natural to set
\begin{equation}\label{E:def psi0}
	\psi_0=R,
\end{equation}
a constant random variable.

\subsection{A multi-step model}\label{subsection:multi-step}

The $n$-step binomial market model with $m>1$ risky assets may be considered as $n$ iterations (not necessarily independent) of the single-step model described in Subsection \ref{subsection:single-step}.
Thus, the natural sample space for an $n$-step model is $\L^n$, the set of all $n$-tuples $(\lambda^1,\dots,\lambda^n)$, where $\lambda^i\in \L$, $i=1,\dots,n$. 

The price of the $i^\text{th}$ asset ($0 \leq i \leq m$) at time $0 \leq k \leq n$ is a random variable $S_i(k) \colon \L^n \to~\RR$
\begin{equation}\label{E:def prices Si}
S_i(k) = S_i(0) \cdot \Psi_i(1) \cdots \Psi_i(k), \qquad i=0,\dots,m
\end{equation}
where $S_i(0)>0$ are constant (the initial price of the $i$-th asset at time $0$) and $\Psi_i(k)$ are the random variables of the price ratio of the $i^\text{th}$ asset ($0 \leq i \leq m$) at time $1 \leq k \leq n$:
\begin{equation}\label{eqn:def PSI_i(k)}
\Psi_i(k) \colon (\lambda^1,\dots,\lambda^n) \mapsto \psi_i(\lambda^k), 
\end{equation}
where $\psi_i$ are defined in \eqref{E:def psii def} (see also \eqref{E:def psii}) and \eqref{E:def psi0}.
Recall that each $S(k)$ is an $\RR^{m+1}$-valued random vector $(S_0(k),S_1(k),\dots,S_m(k))$.
With the terminology in Section \ref{subsection:unform random process} the following observation is straightforward.

\begin{proposition}\label{prop:S(k) uniform quotients}
The $\RR^{m+1}$-valued random vectors $S(0),\dots,S(n)$ form an $n$-step random process on $\L$ with uniform quotient $\phi=(\psi_0,\dots,\psi_m)$.
\end{proposition}

The ``state of the world'' at time $0 \leq k \leq n$ is described by a $k$-tuple $\omega=(\lambda^1,\dots,\lambda^k) \in \L^k$.
More precisely, with the notation \eqref{eqn:bar notation for events} in Section \ref{subsection:martingale measures}, it is the event 
\[
\overline{\omega} = \{\omega\}\times \L^{n-k} 
\]
in $\L^n$.
The ambient space $\L^n$ should be understood from the context. 

Notice that the value of $S_i(k)$ at $\omega \in \L^n$ depends only on the first $k$ entries of $\omega$. 
Hence, the value of $S_i(k)$ only depends on the state of the world at time $k$.
We will therefore abuse notation and regard $S_i(k)$ as a real-valued function on $\L^k$.

\section{Supermodular random variables and a supermodular vertex}\label{section:supermodular}
In this section we single out a certain collection of random variables, first ones that are defined on $\L$ -- the sample space for a single-step market model defined in Subsection \ref{subsection:single-step}, and then ones that are defined on $\L^n$ -- the sample space for the $n$-step market model, see Subsection \ref{subsection:multi-step}.

Recall that the {\em support} of a vector $x \in \RR^k$ is the set $\{1 \leq i \leq k : x_i \neq 0\}$. Denote by $\wp(\{1,\dots,m\})$ the power set (the set of all subsets) of $\{1,\dots,m\}$. Recall that $\L=\{0,1\}^n$. It is straightforward to see that 
\[
\supp \colon \L \xto{ \ \lambda \mapsto \supp(\lambda) \ } \wp(\{1,\dots,m\})
\]
is a bijection, and in this way it equips $\L$ with the partial order
\begin{equation}\label{E:def preceq on L}
\lambda \preceq \lambda' \iff \OP{supp}(\lambda) \subseteq \OP{supp}(\lambda').
\end{equation}
Union and intersection of sets render $\L$ a lattice,  see for example \cite{MR1902334} for background, with {\em join} $\vee$ and {\em meet} $\wedge$ defined by
\begin{eqnarray}
\label{E:def vee wedge}
\lambda \vee \lambda' &=& ( \max\{\lambda(1),\lambda'(1)\}, \dots, \max\{\lambda(m),\lambda'(m)\}) \\
\nonumber
\lambda \wedge \lambda' &=& ( \min\{\lambda(1),\lambda'(1)\}, \dots, \min\{\lambda(m),\lambda'(m)\}).
\end{eqnarray}
Equivalently, this can be defined by
\begin{align*}
\OP{supp}(\lambda \vee \lambda')&=\OP{supp}(\lambda) \cup \OP{supp}(\lambda')
\\
\OP{supp}(\lambda \wedge \lambda')&=\OP{supp}(\lambda) \cap \OP{supp}(\lambda').
\end{align*}


\begin{example}
Let $m=3$ and let $\lambda=(0,1,0)$ and $\lambda'=(1,1,0)$ be two
elements of $\L$. We have the following correspondence:
\begin{eqnarray}
\lambda&=&(0,1,0)\longleftrightarrow \{ 2 \}\equiv \OP{supp}(\lambda) \nonumber\\ 
\lambda'&=&(1,1,0)\longleftrightarrow \{1,2 \}\equiv \OP{supp}(\lambda').\nonumber 
\end{eqnarray}  
Thus, $\lambda \preceq \lambda'$ and $\lambda \vee \lambda' =  \lambda'$ and $\lambda \wedge \lambda'= \lambda$.
\end{example}

\subsection{Supermodular functions on $\L$}

\begin{definition}\label{supermodular_function}
A function $f \colon \L \to \RR$ is called {\em supermodular} if for any $\lambda,\lambda' \in \L$
\begin{equation}\label{D:def_supermodular_fun}
f(\lambda \vee \lambda') + f(\lambda \wedge \lambda') \geq f(\lambda)+f(\lambda').
\end{equation}
It is called {\em modular} if equality holds.
\end{definition}


\begin{proposition}\label{P:supermodular facts}
\hfill
\begin{enumerate}[label=(\roman*)]
\item \label{P:supermodular:linearity}
A linear combination with non-negative coefficients of supermodular functions is supermodular.

\item\label{P:supermodular:affine}
Let $g \colon \RR^m \to \RR$ be an affine function.
Then $g|_{\L}$ is modular.

\item\label{P:supermodular:convex-modular}
Let $g \colon \RR^m \to \RR$ be an affine function of the form $g(x_1,\dots,x_m) = \sum_{i=1}^m a_i x_i + b$, where $a_1,\dots,a_m \geq 0$.
If $h \colon \RR \to \RR$ is convex then the restriction of $h \circ g$ to $\L$ is supermodular.
\end{enumerate}
\end{proposition}

\begin{proof}
Part \ref{P:supermodular:linearity} is \cite[Proposition 2.2.5(a)]{MR3154633}.
Part \ref{P:supermodular:affine} follows from \cite[Theorem 2.2.3]{MR3154633} (and is straightforward).
Part \ref{P:supermodular:convex-modular} follows from \cite[Theorem 2.2.6(a)]{MR3154633}.
\end{proof}

\begin{remark}
In what follows, we will use the term {\em supermodular random variable} as a synonym to {\em supermodular function}.
\end{remark}

Recall the random variables $\psi_i$ \eqref{E:def psii def},\eqref{E:def psi0}; See also \eqref{E:def psii}.

\begin{corollary}\label{C:corollary_supermodular}
Consider a random variable $u \colon \L \to \RR$ of the form
\[
u(\lambda)= h \left( \sum_{i=0}^m \beta_i \cdot \psi_i(\lambda) \right),
\]	
where $h\colon \RR\to \RR$ is a convex function 
and $\beta_1,\dots,\beta_m \geq 0$.
Then $u$ is a supermodular random variable.
\end{corollary}

\begin{proof}
By \eqref{E:def psii} -- \eqref{E:def psi0} 
\[
\sum_{i=0}^m \beta_i \psi_i
=
\beta_0 R + \sum_{i=1}^m \beta_i (D_i + (U_i-D_i)\ell_i
=
\beta_0 R + \sum_{i=1}^m \beta_i D_i + \sum_{i=1}^m \beta_i (U_i-D_i)\ell_i.
\]
Since $\ell_1,\dots,\ell_m \colon \RR^m \to \RR$ introduced in \eqref{E:def elli} are the projections the result follows by applying Proposition \ref{P:supermodular facts}\ref{P:supermodular:convex-modular} with $a_i = \beta_i (U_i-D_i)$ and $b= \beta_0 R + \sum_{i=1}^m \beta_i D_i$.
\end{proof}

\subsection{Fibrewise supermodular functions on $\L^n$}

In this section we fix some $n \geq 0$ and remark that we do not insist that it is the same $n$ as in the Introduction (Section \ref{S:intro}).

\begin{definition}\label{def:fibrewise supermodular}
A function $F \colon \L^n \to \RR$ is called {\em fibrewise supermodular} if for every $1 \leq k \leq n$, every $\omega' \in \L^{k-1}$, and every $\omega'' \in \L^{n-k}$ the function $f \colon \L \to \RR$ defined by
\[
f(\lambda) = F(\omega',\lambda,\omega'')\qquad (\lambda \in \L)
\]
is supermodular.
\end{definition}

Let $F$ be a contingent claim of class $\mathfrak{S}$, see Definition \ref{D:class}.
By definition $F$ is a function of the random variables $S_0(n),S_1(n),\dots,S_m(n)$, the asset prices at time $n$.
As we discussed in Subsection \ref{subsection:multi-step}, these are random variables over the sample space $\L^n$.

\begin{proposition}\label{prop:class S is fibrewise supermodular}
Let $F$ be a contingent claim of class $\mathfrak{S}$.
Then $F$ is a fibrewise supermodular random variable on $\L^n$.
\end{proposition}

\begin{proof}
Let $F$ be as in \eqref{Eq:payoff}.
Fix some $1 \leq k \leq n$, some $\omega'=(\omega^1,\dots,\omega^{k-1}) \in \L^{k-1}$ and $\omega''=(\omega^{k+1},\dots,\omega^n) \in \L^{n-k}$.
For any $1 \leq i \leq m$ set
\[
\alpha_i = S_i(0) \cdot \prod_{j \neq k} \psi_i(\omega^j).
\]
Notice that $\alpha_i \geq 0$ for all $1 \leq i \leq m$ by the assumption that $S_i(0)>0$.
It follows from \eqref{E:def prices Si} that the function $f \colon \L \to \RR$ in Definition \ref{def:fibrewise supermodular} has the form
\[
f(\lambda) = h\left(\sum_{i=0}^m \gamma_i S_i(n)(\omega',\lambda,\omega'')\right)
= h\left( \sum_{i=0}^m  \gamma_i \alpha_i \psi_i(\lambda) \right).
\]
Then $\beta_i = \gamma_i \alpha_i \geq 0$ for all $1 \leq i \leq m$ and Corollary \ref{C:corollary_supermodular} implies that $f$ is supermodular.
\end{proof}

\subsection{The supermodular vertex}\label{sec:supermodular vertex}

We continue studying fibrewise supermodular functions $F \colon \L^n \to \RR$ as in the previous subsection.
Once again we emphasize that in this subsection $n$ need not be the same as the one we fixed in the Introduction (Section \ref{S:intro}).

As in \eqref{E:bi non increasing} let us fix some $b \in \RR^m$ such that $0<b_m \leq \dots \leq b_1 <1$.
We then define $q_0,\dots,q_m$ as in \eqref{E:def qj}.
Observe that the vectors $\rho_0,\dots,\rho_m$ defined in \eqref{E:def rhoj} are all elements of $\L$.
\begin{definition}\label{def:supermodular vertex}
A probability measure $q \in \Delta(\L)$ defined by 
\begin{equation}\label{E:def supervertex}
q(\lambda) =
\left\{
\begin{array}{ll}
q_j & \text{if $\lambda=\rho_j$ for some $0 \leq j \leq m$} \\
0  & \text{otherwise}
\end{array}
\right. ,
\end{equation}
is called a {\em supermodular vertex.}
\end{definition}

It is clear from \eqref{E:def qj} that $q \geq 0$ and that $\sum_{\lambda \in \L} q(\lambda) = \sum_{j=0}^m q_j =1$ so indeed $q \in \Delta(\L)$.
Also, with $\ell_i \colon \L \to \RR$ in \eqref{E:def elli}, the definition of $\rho_j$ in \eqref{E:def rhoj} and of $q_j$ in \eqref{E:def qj} imply
\begin{equation}\label{eqn:E_q(ell)=b}
E_q(\ell_i) = \sum_{j=0}^m \ell_i(\rho_j) q(\rho_j) = \sum_{j=i}^m q_j = b_i.
\end{equation}
Notice that $\ell = (\ell_1,\dots,\ell_m)$ is the inclusion $\ell \colon \L \to \RR^m$ and we have just shown that $E_q(\ell)=b$.
Thus, with the notation of Definition \ref{def:M_n(phi;c)}, and see also \eqref{eqn:M_q(A;c)}, 
\begin{equation}\label{eqn:q in M_q(L;b)}
q \in M_1(\L;b).
\end{equation}
Lemma \ref{lem:tensor M psi c} now shows that for any $n \geq 1$
\[
q^{\otimes n} := \underbrace{q \otimes \dots \otimes q}_{\text{$n$ times}} \in M_n(\L;b).
\]
The fundamental result of \cite{MR4553406} is the next theorem.

\begin{theorem}[{\cite[Theorem A.12]{MR4553406}}]\label{theorem:supermodular vertex maximium}
Consider some $n \geq 1$ and let $f \colon \L^n \to \RR$ be a fibrewise supermodular function.

Suppose $b \in \RR^m$ satisfies $0<b_m \leq \cdots \leq b_1<1$.
Then
\[
\max_{p \in M_n(\L;b)} E_p(f) = E_{q^{\otimes n}}(f),
\]
where $q$ is the supermodular vertex defined in \eqref{E:def supervertex}.
\end{theorem}

\section{Risk neutral measures and calculation of $C_{\max}(F,k)$}\label{section:RN C_max}

In this section we return to the $n$-step discrete time market in the introduction.

\subsection{Risk neutral measures}
Let us keep in mind the formal model we introduced in Section \ref{section:formal model}. Recall that 
$\{S(k)\}_{k=0,\dots,n}$ stands for the $\RR^{m}$-valued stock price process, where $S(k)$ is a vector of stock prices at time $k$: $S(k)=\left(S_1(k),\dots,S_m(k)\right)$. The set of {\em risk-neutral} measures for $\{S(k)\}_{k=0,\dots,n}$ consists of non-degenerate probability measures $p \in \oDelta(\L^n)$ such that for any $0 \leq k < k+ \ell \leq n$ where $k,\ell \geq 0$ and any state of the world $\omega \in \L^k$
\begin{equation}\label{eqn:risk neutral def}
E_p(S_i(k+\ell)| \overline{\omega}) = R^\ell \cdot S_i(k)(\omega), \qquad (1 \leq i \leq m).
\end{equation}
We denote
\[
\OP{RN} = \{ p \in \oDelta(\L^n) : \text{$p$ is risk neutral}\}.
\]
An easy induction argument on $\ell \geq 1$ shows that \eqref{eqn:risk neutral def} is equivalent to the condition that for any $0 \leq k \leq n-1$ and any $\omega \in \L^k$
\[
E_p(S_i(k+1)| \overline{\omega}) = R \cdot  S_i(k)(\omega).
\]
Let us write $R$ for the constant vector $(R,\dots,R) \in \RR^m$. Then  by Definition \ref{def:martingale} and 
in the notation introduced in Section \ref{subsection:unform random process} we have that
\[
\OP{RN}=\oM(S(k);R).
\]
Let $\psi \colon \L \to \RR^m$ denote the random vector $\psi=(\psi_1,\dots,\psi_m)$.
By Propositions \ref{prop:S(k) uniform quotients} and \ref{prop:M_n of uniform}
\begin{equation}\label{eqn:RN=oM_n(psi;r)}
\OP{RN} = \oM_n(\psi;R)
\end{equation}
Recall that $\psi_i = D_i+(U_i-D_i)\ell_i$, see \eqref{E:def psii}, and that $b_i=\tfrac{R-D_i}{U_i-D_i}$, see \eqref{E:def bi}.
Inspection of \eqref{E:def elli} shows that $\ell=(\ell_1,\dots,\ell_m)$ is the inclusion $\L \subseteq \RR^m$.
Inspection of Definition \ref{def:M_n(phi;c)} then shows that
%
\begin{equation}\label{eqn:oM_n(psi;R)=oM_n(L,b)}
\oM_n(\psi;R) = \oM_n(\ell;b) = \oM_n(\L;b).
\end{equation}

\begin{proposition}\label{prop:oM(L,b) not empty}
Consider $n \geq 1$.
Suppose $b \in \RR^m$ satisfies $0<b_1,\dots, b_m<1$.
Then $\oM_n(\L;b)$ is not empty and its closure in $\Delta(\L^n)$ is $M_n(\L;b)$.
\end{proposition}

\begin{proof}
First, we will prove that $\oM_1(\L;b)$ is not empty.
Since $0<b_1,\dots,b_m<1$ we may define 
\[
c_i = \tfrac{b_i}{1-b_i} \qquad (1 \leq i \leq m).
\]
Clearly $c_i>0$.
For any $J \subseteq \{1,\dots ,m\}$ set
\[
c_J= \prod_{j \in J} c_j.
\]
Recall that $\L \approx \wp(\{1,\dots,m\})$ via $\lambda \mapsto \supp(\lambda)$.
Set
\[
C \, \overset{\text{def}}{=} \, \sum_{\lambda \in \L} c_{\supp(\lambda)} = \sum_{J \subseteq \{1,\dots,m\}} c_J = \prod_{j=1}^m (1+c_j).
\]

Define a function $\mu \colon \L \to \RR$ by
\[
\mu(\lambda) = \frac{c_{\supp(\lambda)}}{C}.
\]
Then $\mu \in \Delta(\L)$ since $\mu \geq 0$ and
\[
\sum_{\lambda \in \L} \mu(\lambda) = \tfrac{1}{C} \sum_{\lambda \in \L} c_{\supp(\lambda)} = 1.
\]
Since $c_i>0$ for all $i$ it follows that $\mu>0$, i.e $\mu \in \oDelta(\L)$.
We claim that $\mu \in M_1(\L;b)$.
By \eqref{eqn:M_q(A;c)} we only need to check that $E_{\mu}(\ell_i)=b_i$ for all $1 \leq i \leq m$.
Indeed,
\begin{multline*}
E_{\mu}(\ell_i) = \sum_{\lambda \in \L} \mu(\lambda) \lambda(i)
= \sum_{\lambda \in \L, \lambda(i)=1} \tfrac{c_{\supp(\lambda)}}{C}
= \tfrac{1}{C} \sum_{J \subseteq \{1,\dots,\hat{i},\dots,m\}} c_i c_J
\\
= \tfrac{c_i}{\prod_{j =1}^m (1+c_j)}  \prod_{j \in \{1,\dots,\hat{i},\dots,m\}} (1+c_j) 
= \tfrac{c_i}{1+c_i} = b_i.
\end{multline*}
Now that we have found $\mu \in \oM_1(\L;b)$, Lemma \ref{lem:tensor M psi c} shows that $\mu^{\otimes n} \in \oM_n(\L;b)$.
This proves that $\oM_n(\L;b)$ is not empty.
It remains to show that its closure in $\Delta(\L^n)$ is $M_n(\L;b)$.

Consider some $p \in M_n(\L;b)$.
Since $M_n(\L;b)$ is a polytope, it is a convex set, so for any $0 \leq t \leq 1$,
\[
z_t:=(1-t) \cdot p + t \cdot \mu^{\otimes n} \in M_n(\L;b).
\]
Clearly $z_0=p$ and $z_t>0$ for any $t>0$ since $\mu^{\otimes n}>0$.
Thus $z_t \in \oM_n(\L;b)$ for all $0<t \leq 1$.
It follows that $p$ is in the closure of $\oM_n(\L;b)$.
This shows that $M_n(\L;b)$ is contained in the closure of $\oM_n(\L;b)$ and the reverse inclusion is clear because $M_n(\L;b)$ is compact.
\end{proof}

\subsection{Pricing contingent claims of class $\mathfrak{S}$ in the multi-step model.}\label{SS:c-max}
Let $F$ be a 
contingent claim in the $n$-step model from Section \ref{S:intro}.
A {\em no-arbitrage price} of $F$ at time $k=0,1,\dots , n-1$ and at the state of the world $\omega \in \L^k$ is defined as follows.
\begin{equation}\label{D:def C n step}
	C_p(F,k)(\omega)= R^{k-n} \cdot E_p(F|\overline{\omega}),
\end{equation}
where $p$ is a risk-neutral measure, i.e $p\in \OP{RN}$. 
The no-arbitrage price $C_p(F,k)$ depends on the state of the world at time $k$, i.e some $\omega \in \L^k$. 

Let us fix a state of the world $\omega \in \L^k$ at time $k$. 
If $m \geq 2$ the model is incomplete, so by varying the risk-neutral measure $p$, the no-arbitrage prices $C_p(F,k)(\omega)$ form an open interval
\begin{equation}\label{D:def price interval n}
	\left(C_{\min}(F,k)(\omega),C_{\max}(F,k)(\omega)\right)=\left\{ C_p(F,k)(\omega)\in \RR\ |\ p\in \OP{RN} \right\}.
\end{equation}
If $m=1$ (i.e the model is complete) or if $k=n$ this open interval degenerates to a closed interval consisting of a  single point.
We are interested in the upper bound $C_{\max}(F,k)(\omega)$ of this interval for given time $0 \leq k \leq n-1$ and state of the world $\omega \in \L^k$,
\begin{equation}\label{eq:Cmax sup formula}
C_{\max}(F,k)(\omega) = \sup_{p \in \OP{RN}} C_p(F,k)(\omega).
\end{equation}

We arrive at the main result of this subsection.
Recall that $F$ in class $\mathfrak{S}$ is a random variable on $\L^n$.
With the notation in Subsection \ref{subsec:product and conditional prob}, if $\omega \in \L^k$ then  $F(\omega,-)$ is a random variable on $\L^{n-k}$. Recall the definition of the supermodular vertex, a special risk-neutral measure $q$ (see Subsection \ref{sec:supermodular vertex}).

\begin{theorem}\label{theorem:C_max with supervertex}
Let $F$ be a contingent claim of class $\mathfrak{S}$ in the $n$-step market model, where $n\geq 1$.	
Then for any state of the world $\omega\in \L^k$ at time $0 \leq k \leq n-1$
\[
C_{\max}(F,k)(\omega) = R^{k-n} \cdot E_{q^{\otimes n-k}}(F(\omega-)).
\]
\end{theorem}

\begin{proof}
Equations \eqref{eqn:RN=oM_n(psi;r)}, \eqref{eqn:oM_n(psi;R)=oM_n(L,b)} and \eqref{eq:Cmax sup formula}, together with \eqref{D:def C n step} and Lemma \ref{lem:products conditional general}\ref{lem:products conditional general:box restriction} show that
\[
C_{\max}(F,k)(\omega) = \sup_{p \in \oM_n(\L;b)} \, R^{k-n} \cdot E_p(F|\overline{\omega}) 
 = \sup_{p \in \oM_n(\L;b)} \, R^{k-n} \cdot E_{p_{|\overline{\omega}}}(F(\omega-)),
\]
where we view $p_{|\overline{\omega}}$ as a probability 
measure on $\overline{\omega} \approx \L^{n-k}$.
Notice that if $k=0$ then $\omega$ is the empty word, so $p_{|\overline{\omega}}=p$ and $\oM_n(\L;b) \xto{p \mapsto p_{|\overline{\omega}}} \oM_{n-k}(\L;b)$ is the identity map.
Proposition \ref{prop:oM(L,b) not empty} and Corollary \ref{cor:oM(L,b) conditional surjective} now show that for any $0 \leq k \leq n-1$
\[
\sup_{p \in \oM_n(\L;b)} E_{p_{|\overline{\omega}}}(F(\omega-)) =
\sup_{p \in \oM_{n-k}(\L;b)} E_{p}(F(\omega-)).
\]
Since $p \mapsto E_p(X)$ is a continuous function, Proposition \ref{prop:oM(L,b) not empty} and the fact that $M_{n-k}(\L;b)$ is a compact subset of $\Delta(\L^{n-k})$  imply that
\[
\sup_{p \in \oM_{n-k}(\L;b)} E_{p}(F(\omega-)) =
\max_{p \in M_{n-k}(\L;b)} E_{p}(F(\omega-)).
\]
By Proposition \ref{prop:class S is fibrewise supermodular} $F$ is a fibrewise supermodular function $F \colon \L^n \to \RR$.
But then it is immediate from the definition of fibrewise supermodularity that $F(\omega-) \colon \L^{n-k} \to \RR$ is also fibrewise supermodular.
We apply Theorem \ref{theorem:supermodular vertex maximium} to obtain
\[
\max_{p \in M_{n-k}(\L;b)} E_{p}(F(\omega-)) = E_{q^{\otimes n-k}}(F(\omega-)).
\]
The result follows.
\end{proof}

\begin{proof}[Proof of Proposition \ref{P:Cmax formula}]
We assume $F$ has the form in \eqref{Eq:payoff}.
By Theorem \ref{theorem:C_max with supervertex}
\[
C_{\max}(F,k)(\omega) = 
R^{k-n} \cdot E_{q^{\otimes n-k}}(F(\omega-)) 
\]
Recall the definition of $q \in \Delta(\L)$ in \eqref{E:def supervertex} and that $q^{\otimes n-k}$ is the resulting  product measure on $\L^{n-k}$.
With the notation in \eqref{eqn:Pkm}, \eqref{E:def qJ} and \eqref{E:def chiJ} and since $S_i(k)$ is given by \eqref{eqn:S_(k) formula with psi}, 
\begin{align*}
E_{q^{\otimes n-k}}(F(\omega-)) 
&=
\sum_{J \in \C P_{n-k}(m)} q^{\otimes n-k}(\rho_{j_1},\dots,\rho_{j_{n-k}}) \cdot h\left( \sum_{i=0}^m \gamma_i \cdot S_i(k)(\omega,\rho_{j_1},\dots,\rho_{j_{n-k}})\right)
\\
&=
\sum_{J \in \C P_{n-k}(m)} q_J  \cdot h\left( \sum_{i=0}^m \gamma_i \cdot S_i(k)(\omega) \cdot \chi_i(J) \right).
\end{align*}
The result follows.
\end{proof}

\section{Minimum cost super-hedging strategy}\label{S:proof}

Let $F$ be a contingent claim of class $\mathfrak{S}$ in our $n$-step binomial model. 
A hedging strategy $\beta$ in this market consists of random variables $\beta_i(k)$ on $\L^n$ where $0 \leq i \leq m$ and $0 \leq k \leq n-1$ such that $\beta_i(k)$ depend only on the state of the world at time $k$.
Thus, we view $\beta_i(k)$ as real valued functions from $\L^k$.
They give rise to a hedging portfolio $V_\beta(k)$ whose set up cost at the state of the world $\omega$ at time $k$ is
\[
V_\beta(k)(\omega) = \sum_{i=0}^m \beta_i(k)(\omega) \cdot S_i(k)(\omega).
\]
In light of \eqref{E:def prices Si} and \eqref{eqn:def PSI_i(k)} the value of this portfolio at a subsequent state of the world $\omega \lambda$ at time $k+1$, where $\lambda \in \L$, is
\begin{equation}\label{eqn:V_beta+1}
V_\beta^{+1}(k)(\omega\lambda) = \sum_{i=0}^m \beta_i(k)(\omega) \cdot S_i(k+1)(\omega \lambda) 
=
\sum_{i=0}^m \beta_i(k)(\omega) \cdot S_i(k)(\omega) \cdot \psi_i(\lambda). 
\end{equation}
Given a state of the world $\omega \in \L^k$ at time $k$ we obtain a random variable $\Phi_{\beta,\omega}$ on $\L$ 
\begin{equation}\label{eq:def Phi_beta,omega}
\Phi_{\beta,\omega}(\lambda) = V_\beta^{+1}(k)(\omega\lambda), \qquad (\lambda \in \L).
\end{equation}
Recall the definition of a special risk-neutral measure $q$, the supermodular vertex (see Section \ref{sec:supermodular vertex} ).
By \eqref{eqn:E_q(ell)=b}, \eqref{E:def psii}, \eqref{E:def bi} and the linearity of the expectation, $E_q(\psi_i)=R$.
Hence
\begin{equation}\label{eq:Eqv+1}
E_q(\Phi_{\beta,\omega})=
\sum_{i=0}^m \beta_i(k)(\omega) \cdot S_i(k)(\omega) \cdot E_q(\psi_i)
=
R \cdot V_\beta(k)(\omega). 
\end{equation}
\begin{proof}[Proof of Theorem \ref{T:minimum cost hedgiging theorem}]
Let $\beta$ be a hedging strategy with associated portfolio $V_\beta$.
By Definition \ref{defn:hedging strategy}, $\beta$ is a super-hedge if for every state of the world $\omega$ at time $0 \leq k \leq n-1$ 
\begin{equation}\label{eq:beta suer hedge}
V_\beta^{+1}(k)(\omega \lambda) \geq C_{\max}(F,k+1)(\omega \lambda) \qquad \text{for every $\lambda \in \L$.}
\end{equation}

The purpose of this theorem is to find a super-hedging strategy $\alpha$ whose setup cost at any time $k$  is no larger than the setup cost of any other super-hedging strategy $\beta$, i.e
\[
V_\alpha(k) \leq V_\beta(k)  \qquad \text{for any $0 \leq k \leq n-1$},
\]
and show that, in fact, at any time $0 \leq k \leq n-1$ its setup cost is equal to $C_{\max}(F,k)$, i.e
\[
V_\alpha(k)(\omega) = C_{\max}(F,k)(\omega) \qquad \text{for any state of the world $\omega \in \L^k$.}
\]

{\bf Claim 1:} Consider a state of the world $\omega \in \L^k$ for some $0 \leq k \leq n-1$. Recall formula \eqref{E:def rhoj} which defines special sample space elements $\rho_j \in \L$, $0 \leq j \leq m$.

The system of $m+1$ linear equations (one for each $0 \leq j \leq m$) in the $m+1$ unknowns $\alpha_0,\dots,\alpha_m$
\begin{equation}\label{eq:m+1 equations alpha}
  \sum_{i=0}^m \alpha_i \cdot S_i(k)(\omega) \cdot \psi_i(\rho_j) = C_{\max}(F,k+1)(\omega\rho_j), \qquad (0 \leq j \leq m)
\end{equation}
has a unique solution.
In fact, if we write $\alpha_*$ for the column vector $(\alpha_i)_{i=0}^m$ and $c_*$ for the column vector $(C_{\max}(F,k+1)(\omega\rho_j))_{j=0}^m$ then the system \eqref{eq:m+1 equations alpha} has the form
\[
M \alpha_* = c_*
\]
where $M=Q^{-1}\cdot N \cdot T(\omega)$ for the matrices $Q,N,T(\omega)$ defined in the statement of the theorem.

{\bf Proof:}
By inspection of \eqref{eq:m+1 equations alpha} $M$ is the $(m+1)\times(m+1)$ matrix with entries $M_{j,i}=S_i(k)(\omega) \cdot \psi_i(\rho_j)$ where $0 \leq i,j \leq m$.
It follows from \eqref{E:def psii def} and \eqref{E:def psi0} and from \eqref{E:def chii} that $\psi_i(\rho_j)=\chi_i(j)$.
By inspection $M$ can be written as the product 
\begin{equation}
M=
\underbrace{
\begin{pmatrix*}[r]
1 & \chi_1(0) & \cdots & \chi_m(0) \\
1 & \chi_1(1) & \cdots & \chi_m(1) \\
\vdots & \vdots &      & \vdots \\
1 & \chi_1(m) & \cdots & \chi_m(m)
\end{pmatrix*}
}_{M'} 
\cdot 
\underbrace{
\begin{pmatrix*}[r]
RS_0(k)(\omega) &  \\
                & S_1(k)(\omega) \\
                &                & \ddots \\
                &                &        & S_m(k)(\omega)
\end{pmatrix*}
}_{T(\omega)}.
\end{equation}
Notice that $T(\omega)$ is invertible since $R, S_i(k)(\omega)>0$. 

It follows from \eqref{E:def chii} that for any $1 \leq j \leq m$ and any $0 \leq i \leq m$
\[
\chi_i(j) - \chi_i(j-1) = 
\left\{
\begin{array}{ll}
	0 & i \neq j \\
	U_i-D_i & i=j
\end{array}
\right.
\]
With this in mind, and with $Q,N$ defined in \eqref{D:def_matrices}, a straightforward calculation gives
\[
QM'=N.
\]
Both $Q$ and $N$ are invertible since they are upper/lower triangular matrices with non-zero entries on the diagonal.
It follows that $M=Q^{-1}NT(\omega)$ is invertible.
In particular the system $M\alpha_*=c_*$ has a unique solution $\alpha_* = M^{-1} c_*$.
This completes the proof of Claim 1.

Claim 1 allows us to define random variables $\alpha_i(k)$ on $\L^n$, namely functions $\alpha_i(k) \colon \L^n \to \RR$, for $0 \leq i \leq m$ and $0 \leq k \leq n-1$.
Specifically, for any $0 \leq i \leq m$, any $0 \leq k \leq n-1$ and any $\omega \in \L^k$ let the vector $\alpha(k)(\omega) \in \RR^{m+1}$ be the unique solution of \eqref{eq:m+1 equations alpha}.
That is, for any $0 \leq k \leq n-1$ and any $\omega \in \L^k$
\[
(\alpha_i(k)(\omega))_{i=0}^m = T(\omega)^{-1} N^{-1} Q \cdot (C_{\max}(F,k+1)(\omega \rho_j))_{j=0}^m
\]
This defines a hedging strategy $\alpha$ with associated portfolio $V_\alpha(k)$.
Since the left hand side of the $j^\text{th}$ equation in \eqref{eq:m+1 equations alpha} is $V^{+1}_\alpha(k)(\omega \rho_j)$, we see that this hedging portfolio satisfies
\begin{equation}\label{eq:Valpha+1 C_max}
V^{+1}_\alpha(k)(\omega \rho_j) = C_{\max}(F,k+1)(\omega \rho_j)
\end{equation}
for any state of the world $\omega \in \L^k$ at time $0 \leq k \leq n-1$ and any $0 \leq j \leq m$.

{\bf Claim 2:} $\alpha_i(k)$ define a super-hedging strategy.
Moreover, the value of the portfolio $V_\alpha(k)$ at time $k$ is equal to $C_{\max}(F,k)$, i.e $V_\alpha(k)(\omega)=C_{\max}(F,k)(\omega)$ for all $0 \leq k \leq n-1$ and all $\omega \in \L^k$.

{\bf Proof:}
To show that $\alpha_i(k)$ define a super-hedge we need to show that for any state of the world $\omega$ at time $k$ and any $\lambda \in \L$,
\begin{equation}\label{eq:alpha superhedge check}
\underbrace{V^{+1}_\alpha(k)(\omega\lambda)}_{\Phi_{\alpha,\omega}(\lambda)} \geq \underbrace{C_{\max}(F,k+1)(\omega\lambda)}_{\Xi_\omega(\lambda)}, \qquad (\lambda \in \L).
\end{equation}
The notation for the left hand side coincides with \eqref{eq:def Phi_beta,omega}.
For the right hand side, notice that $\Xi_\omega$ is a random variable on $\L$.
Suppose that \eqref{eq:alpha superhedge check} fails.
Among all $\lambda \in \L$ for which $\Phi_{\omega,\alpha}(\lambda)<\Xi_\omega(\lambda)$ choose one with maximal possible $j$ such that $\rho_j \preceq \lambda$ (see \eqref{E:def preceq on L} and \eqref{E:def rhoj}).
If $j=m$ then $\rho_m \preceq \lambda$ implies that $\lambda=\rho_m$ because $\rho_m$ is the maximum element $(1,\dots,1)$ in $\L$.
This contradicts \eqref{eq:Valpha+1 C_max} which reads $\Phi_{\alpha,\omega}(\rho_m)=\Xi_\omega(\rho_m)$.
We deduce that $j<m$.

Set $\lambda'=\lambda \vee \rho_{j+1}$ (see \eqref{E:def vee wedge}).
By the maximality of $j$ we get that $\lambda \wedge \rho_{j+1}=\rho_j$.
It is clear from \eqref{eq:def Phi_beta,omega}, \eqref{eqn:V_beta+1} and from \eqref{E:def psii} that $\Phi_{\omega,\alpha} \colon \L \to \RR$ is an affine function, hence it is modular by Proposition \ref{P:supermodular facts}\ref{P:supermodular:affine}. It follows from Definition \ref{supermodular_function} that
\begin{equation}\label{eqn:Phi equality}
\Phi_{\omega,\alpha}(\lambda')+\Phi_{\omega,\alpha}(\rho_j) = \Phi_{\omega,\alpha}(\lambda) + \Phi_{\omega,\alpha}(\rho_{j+1}).
\end{equation}
It follows from Proposition \ref{P:Cmax formula} and from \eqref{E:def prices Si} that
\[
C_{\max}(F,k+1)(\omega\lambda) =
R^{k+1-n} \sum_{J \in \C P_{n-k-1}(m)} q_J\cdot h\left( \sum_{i=0}^m \gamma_i \cdot \chi_i(J) \cdot S_i(k)(\omega)\cdot \psi_i(\lambda) \right).
\]
By Proposition \ref{P:supermodular facts}\ref{P:supermodular:convex-modular} the functions $\lambda \mapsto h\left( \sum_{i=0}^m \gamma_i \cdot \chi_i(J) \cdot S_i(k)(\omega)\cdot \psi_i(\lambda) \right)$ are supermodular.
By part \ref{P:supermodular:linearity} of that proposition, the function  $\lambda \mapsto C_{\max}(F,k+1)(\omega\lambda)$ is supermodular.
That is, $\Xi_\omega$ is supermodular.
Therefore, see \eqref{D:def_supermodular_fun},
\begin{equation}\label{eqn:Xi inequality}
\Xi_\omega(\lambda')+\Xi_\omega(\rho_j) \geq \Xi_\omega(\lambda) + \Xi_\omega(\rho_{j+1})
\end{equation}
By \eqref{eq:Valpha+1 C_max}, $\Phi_{\omega,\alpha}(\rho_i)=\Xi(\rho_i)$ for all $0 \leq i \leq m$.
Subtracting \eqref{eqn:Phi equality} from \eqref{eqn:Xi inequality} yields
\[
\Xi_\omega(\lambda') - \Phi_{\omega,\alpha}(\lambda') \geq  \Xi_\omega(\lambda)-\Phi_{\omega,\alpha}(\lambda). 
\]
By assumption on $\lambda$, $\Phi_{\omega,\alpha}(\lambda) < \Xi_\omega(\lambda)$.
Therefore $\Phi_{\omega,\alpha}(\lambda')<\Xi_\omega(\lambda')$.
But by definition $\rho_{j+1} \preceq \lambda'$, which contradicts the maximality of $j$ in the choice of $\lambda$.
This establishes \eqref{eq:alpha superhedge check}, i.e that $\alpha$ is a super-hedging strategy.

It remains to prove that $V_\alpha(k)(\omega)=C_{\max}(F,k)(\omega)$.
By \eqref{eq:Valpha+1 C_max} $\Phi_{\alpha,\omega}(\rho_j)=\Xi_\omega(\rho_j)$ for all $0 \leq j \leq m$.
Since $q \in M_1(\L;b)$, see \eqref{eqn:q in M_q(L;b)}, and $q$ is supported on $\{\rho_0,\dots,\rho_m\}$ we get from \eqref{eq:Eqv+1}
\[
R \cdot V_{\alpha}(k)(\omega) 
=
E_q(\Phi_{\omega,\alpha}) 
=
E_q(\Xi_{\omega}).
\]
It follows from Theorem \ref{theorem:C_max with supervertex} and from Lemma \ref{lem:products conditional general}\ref{lem:products conditional general:E in stages} that
\begin{align}\label{eq:EqXi}
E_q(\Xi_{\omega}) 
&=
E_q(\lambda \mapsto C_{\max}(F,k+1)(\omega\lambda))
\\
\nonumber
&=R^{k+1-n} \cdot E_q(\lambda \mapsto E_{q^{\otimes n-k-1}}(F(\omega\lambda-)))
\\
\nonumber
&=
R^{k+1-n} \cdot E_{q^{\otimes n-k}}(F(\omega-))
\\
\nonumber
&=R \cdot C_{\max}(F,k)(\omega).
\end{align}
We deduce that $R \cdot V_\alpha(k)(\omega) = R \cdot C_{\max}(F,k)(\omega)$.
Since $R>0$ it follows that $V_\alpha(k)(\omega) = C_{\max}(F,k)(\omega)$.
This completes the proof of Claim 2.

To complete the proof it remain to show that $\alpha$ is a minimum super-hedging strategy.
Let $\beta$ be an arbitrary super-hedging strategy, i.e \eqref{eq:beta suer hedge} holds for $\beta$.
This can be written using \eqref{eq:def Phi_beta,omega} and the right hand side of \eqref{eq:alpha superhedge check} as the following inequality for  any $0 \leq k \leq n-1$ and any $\omega \in \L^k$
\[
\Phi_{\beta,\omega}(\lambda) \geq \Xi_{\omega}(\lambda), \qquad \lambda \in \L.
\]
By \eqref{eq:Eqv+1}, the monotonicity of the expectation, and \eqref{eq:EqXi}
\[
R \cdot V_\beta(k)(\omega) =
E_q(\Phi_{\beta,\omega}) \geq
E_q(\Xi_\omega) 
=
R \cdot C_{\max}(F,k)(\omega).
\]
Together with Claim 2 we get $V_\beta(k)(\omega) \geq C_{\max}(F,K)(\omega)=V_\alpha(k)(\omega)$.
Thus, $\alpha$ is a minimum cost super-hedging strategy.
\end{proof}

\section{Residuals of a minimum cost super-hedging strategy.}\label{S:Residuals}
Let $\alpha(k)=(\alpha_0(k),\alpha_1(k),\ldots,\alpha_m(k))$, for $k=0,1,\ldots, n-1$ be a minimum cost super-hedging strategy for a contingent claim $F$ of class $\mathfrak{S}$ in our $n$-step binomial model. 

At time $k=0$, an investor with a short position in contingent claim $F$ purchases a  super-hedging portfolio which consists of assets $S_0,\dots , S_m$ using the amount $C_{\max}(F,0)$ to cover a setup cost $V_\alpha(0)$:  
\[
V_\alpha(0) = \sum_{i=0}^m \alpha_i(0) S_i(0)=C_{\max}(F,0).
\]
At time $k=1$, the asset prices change, and the value of the super-hedging portfolio becomes
\[
\sum_{i=0}^m \alpha_i(0) S_i(1).
\] 
By construction, we have
\[
\sum_{i=0}^m \alpha_i(0) S_i(1)\geq C_{\max}(F,1).
\] 
At $k=1$, the super-hedging portfolio is liquidated, and a new portfolio is purchased. The required setup cost for the new super-hedging portfolio at time $k=1$ is $C_{\max}(F,1)$. It follows from the above inequality that there exists a state of the world $\omega \in \L$ such that 
\[
\sum_{i=0}^m \alpha_i(0) S_i(1)(\omega) > C_{\max}(F,1)(\omega).
\] 
It means that at the state of the world $\omega$ at time $k=1$, the super-hedging strategy produces a positive local residual amount, the surplus between the liquidation value of the super-hedging portfolio and the required setup cost of the next super-hedging portfolio. This residual amount has to be withdrawn before the next hedging portfolio is set up. Let us generalize this discussion as follows. 
 
\begin{definition}\label{D:residual}
	A {\em local residual} of a super-hedging strategy $\alpha(k)$ at time $k$ is denoted by $\delta_{\alpha}(k)$, $k=1,\ldots,n$ and is defined as follows:
	\begin{equation}\label{Eq:residual}
		\delta_{\alpha}(k) = \sum_{i=0}^m\alpha_i(k-1)S_i(k) - C_{\max}(F,k), \end{equation} 
where $C_{\max}(F,n)\equiv F$.	 
\end{definition}
\begin{remark}\label{R:maturity}
 We impose a condition that the terminal value of the super-hedging portfolio (at $k=n$) matches the liability $F$.
\end{remark}
It follows from \eqref{E:hedge condition} that the local residuals of a super-hedging strategy are always
non-negative. Moreover, in our incomplete model for some scenarios the
local residuals are strictly positive. In such cases, the hedging portfolio
liquidation value at time $k$ exceeds the necessary setup cost of the next portfolio built at time $k$ (for $k=1,\dots, n-1$), or exceeds the contingent claim terminal value (for $k=n$).

It follows from the above discussion that a super-hedging strategy is {\em non-self-financing}. The local residuals $\delta_{\alpha}(k)$, for $k=1,\dots, n$ are withdrawn and deposited at the risk-free rate. At the contingent claim maturity time $k=n$, the {\em accumulated residual} amount of the super-hedging strategy $\alpha(k)$ is
\begin{equation}\label{Eq:accum_residual}
	\Delta_{\alpha} = \sum_{i=1}^n \delta_{\alpha}(k)R^{n-k}.
 \end{equation} 
Observe that, since $C_{\max}(F,k)$ is an arbitrage price of a contingent claim $F$ at time $k=0,\dots,n-1$, the minimum cost super-hedge $\alpha(k)$ is
an {\em arbitrage strategy}.  

\begin{remark}\label{R:explicit_residual}
	Notice that both the upper bound $C_{\max}(F,k)$ and the super-hedging strategy $\alpha(k)$ are given by explicit (non-recursive) formulas (see  \eqref{E:Cmax formula} and \eqref{E:solution}, respectively). Therefore, \eqref{Eq:residual} provides an explicit representation for the local residuals $\delta(k)$ for $k=1,\dots,n$ in terms of the model parameters and the asset price paths up to time $k$.
\end{remark}

\section{An example}\label{S:example}
The purpose of this section is to evaluate somewhat abstract formulas of the paper, first, in a simple general example, and then in a concrete numerical situation.

\subsection{Evaluation of formulas}
Consider an example of a $2$-step model with $m=2$ risky assets. Recall that
\begin{align*}
b_i &= \frac{R-D_i}{U_i-D_i} \text{ and we assume that } b_1\geq b_2.\\
q_0 &=1-b_1,\ q_1=b_1-b_2,\ q_2=b_2,\\
q_{(0,0)} &= (1-b_1)^2,\ q_{(0,1)}=q_{(1,0)}=(1-b_1)(b_1-b_2),\ q_{(0,2)}=q_{(2,0)}=(1-b_1)b_2,\\
q_{(1,1)} &= (b_1-b_2)^2,\ q_{(1,2)}=q_{(2,1)}=(b_1-b_2)b_2,\ q_{(2,2)}=b_2^2,\\
\chi_0(j)&=R,\ \chi_1(0)=D_1,\ \chi_1(1)=\chi_1(2)=U_1,\ \chi_2(0)=\chi_2(1)=D_2,\ \chi_2(2)=U_2,\\
\chi_0(J)&=R^2,\ \chi_1(0,0)=D_1^2,\chi_1(0,1)=\chi_1(1,0)=\chi_1(0,2)=\chi_1(2,0)=D_1U_1,\\
\chi_1(1,1)&=\chi_1(1,2)=\chi_1(2,1)=\chi_1(2,2)=U_1^2,\\
\chi_2(0,0)&=\chi_2(0,1)=\chi_2(1,0)=\chi_2(1,1)=D_2^2,\\ 
\chi_2(0,2)&=\chi_2(2,0)=\chi_2(1,2)=\chi_2(2,1)=D_2U_2,\ \chi_2(2,2)=U_2^2.
\end{align*}
Since the model is $2$-step, we need to compute only $C_{\max}(F,0)$ and $C_{\max}(F,1)$,
while $C_{\max}(F,2)=F$.
In the following computation we identify $\C P_1(2)$ with $\{0,1,2\}$.
\begin{align}
C_{\max}(F,1)
=&\ R^{-1}\sum_{j=0}^2 q_j\cdot h\left(\sum_{i=0}^2 \gamma_i S_i(1)\chi_i(j)\right)\nonumber \\
=&\ R^{-1} 
\Big( 
(1-b_1)\cdot h\left(\gamma_0 S_0(1)R + \gamma_1 S_1(1) D_1 + \gamma_2 S_2(1)D_2\right)\label{Eq:Cmax1}\\
&+ (b_1-b_2)\cdot h\left(\gamma_0 S_0(1)R + \gamma_1 S_1(1) U_1 + \gamma_2 S_2(1)D_2\right)\nonumber \\
& + b_2\cdot h\left(\gamma_0 S_0(1)R + \gamma_1 S_1(1) U_1 + \gamma_2 S_2(1)U_2\right)
\Big)\nonumber
\end{align}
Recall that
$
\rho_0=\left( 
\begin{smallmatrix}
0\\
0
\end{smallmatrix}
 \right),\
\rho_1=\left( 
\begin{smallmatrix}
1\\
0
\end{smallmatrix}
 \right),\
\rho_2=\left( 
\begin{smallmatrix}
1\\
1
\end{smallmatrix}
 \right)
$
and hence we get that
\begin{align*}
C_{\max}(F,1)(\rho_0) 
=&\ R^{-1}\Big( 
(1-b_1)\cdot h\left( \gamma_0 S_0(0)R^2+\gamma_1 S_1(0)D_1^2 + \gamma_2 S_2(0)D_2^2 \right)\\
&+ (b_1-b_2)\cdot h\left( \gamma_0 S_0(0)R^2 + \gamma_1 S_1(0)D_1U_1 + \gamma_2 S_2(0)D_2^2 \right)\\
&+ b_2\cdot h\left( \gamma_0 S_0(0)R^2 + \gamma_1 S_1(0)D_1U_1 + \gamma_2 S_2(0)D_2U_2 \right)\Big),
\end{align*}
\begin{align*}
C_{\max}(F,1)(\rho_1) 
=&\ R^{-1} 
\Big( 
(1-b_1)\cdot h\left(\gamma_0 S_0(1)R^2 + \gamma_1 S_1(0) D_1U_1 + \gamma_2 S_2(0)D_2^2\right)\\
&+ (b_1-b_2)\cdot h\left(\gamma_0 S_0(0)R^2 + \gamma_1 S_1(0) U_1^2 + \gamma_2 S_2(0)D_2^2\right)\\
& + b_2\cdot h\left(\gamma_0 S_0(0)R^2 + \gamma_1 S_1(0) U_1^2 + \gamma_2 S_2(0)D_2U_2\right)
\Big),
\end{align*}
\begin{align*}
C_{\max}(F,1)(\rho_2) 
=&\ R^{-1} 
\Big( 
(1-b_1)\cdot h\left(\gamma_0 S_0(1)R^2 + \gamma_1 S_1(0) D_1U_1 + \gamma_2 S_2(0)D_2U_2\right)\\
&+ (b_1-b_2)\cdot h\left(\gamma_0 S_0(0)R^2 + \gamma_1 S_1(0) U_1^2 + \gamma_2 S_2(0)D_2U_2\right)\\
& + b_2\cdot h\left(\gamma_0 S_0(0)R^2 + \gamma_1 S_1(0) U_1^2 + \gamma_2 S_2(0)U_2^2\right)
\Big).
\end{align*}
Since some terms of $C_{\max}(F,0)$ are equal for certain values of $J\in \C P_2(2)$,
the computation simplifies which is already visible in this simple example.
\begin{align*}
C_{\max}(F,0)
=&\ R^{-2}\sum_{J\in \C P_2(2)} q_J\cdot h\left( \sum_{i=0}^2 \gamma_i S_i(0)\chi_i(J) \right)\\
=&\ R^{-2}\Big(
q_{(0,0)}\cdot h\left( \gamma_0 S_0(0)R^2 + \gamma_1 S_1(0) D_1^2 + S_2(0)D_2^2 \right)\\
&+ 2q_{(0,1)}\cdot h\left( \gamma_0 S_0(0)R^2 + \gamma_1 S_1(0) D_1U_1 + S_2(0)D_2^2 \right)\\
&+ 2q_{(0,2)}\cdot h\left( \gamma_0 S_0(0)R^2 + \gamma_1 S_1(0) D_1U_1 + S_2(0)D_2U_2 \right)\\
&+ q_{(1,1)}\cdot h\left( \gamma_0 S_0(0)R^2 + \gamma_1 S_1(0) U_1^2 + S_2(0)D_2^2 \right)\\
&+ 2q_{(1,2)}\cdot h\left( \gamma_0 S_0(0)R^2 + \gamma_1 S_1(0) U_1^2 + S_2(0)D_2U_2 \right)\\
&+ q_{(2,2)}\cdot h\left( \gamma_0 S_0(0)R^2 + \gamma_1 S_1(0) U_1^2 + S_2(0)U_2^2 \right)\Big)\\
\end{align*}

\subsection{A numerical example}\label{SS:example}
Consider the following concrete case with the following parameters
\begin{align*}
S_0(0)&= S_1(0)=S_2(0)=100,\\
R=1,\ D_1&=0.8,\ D_2=0.9,\ U_1=1.1,\ U_2=1.2.
\end{align*}
Assume that $\gamma_0=-1,\ \gamma_1=\gamma_2=\frac{1}{2}$ and $h(x) = \max\{x,0\}$.
In other words, $F$ is a European basket call with strike price $K=100$.

The sample space $\C L^2$ can be conveniently identified with the
set of $(2\times 2)$-matrices with binary entries. The columns
correspond to the price ratios of both assets at a given time, while the rows
correspond to the price dynamics of individual assets. In particular, we have
\begin{align*}
S_1(1)
\left( 
\begin{smallmatrix}
0 & *\\
* & *
\end{smallmatrix}
 \right) & = 80,\
S_1(1)
\left( 
\begin{smallmatrix}
1 & *\\
* & *
\end{smallmatrix}
 \right)  = 110,\\
S_1(2)
\left( 
\begin{smallmatrix}
0 & 0\\
* & *
\end{smallmatrix}
 \right)  &= 64,\
S_1(2)
\left( 
\begin{smallmatrix}
0 & 1\\
* & *
\end{smallmatrix}
 \right) 
=
S_1(2)
\left( 
\begin{smallmatrix}
1 & 0\\
* & *
\end{smallmatrix}
 \right)  = 88,\
S_1(2)
\left( 
\begin{smallmatrix}
1 & 1\\
* & *
\end{smallmatrix}
 \right)  = 121,\\
S_2(1)
\left( 
\begin{smallmatrix}
* & *\\
0 & *
\end{smallmatrix}
 \right) & = 90,\
S_2(1)
\left( 
\begin{smallmatrix}
* & *\\
1 & *
\end{smallmatrix}
 \right)  = 120,\\
S_2(2)
\left( 
\begin{smallmatrix}
* & *\\
0 & 0
\end{smallmatrix}
 \right)  &= 81,\
S_1(2)
\left( 
\begin{smallmatrix}
* & *\\
0 & 1
\end{smallmatrix}
 \right) 
=
S_2(2)
\left( 
\begin{smallmatrix}
* & *\\
1 & 0
\end{smallmatrix}
 \right)  =108,\
S_2(2)
\left( 
\begin{smallmatrix}
* & *\\
1 & 1
\end{smallmatrix}
 \right)  = 144.
\end{align*}
Next we have $b_1=\frac{2}{3}$ and $b_2=\frac{1}{3}$, so $b_1\geq b_2$ and we do not need to
reorder the assets. The matrices $Q,\ N,$ and $T$ are given by
$$
Q=
\begin{pmatrix*}[r]
1 & 0 & 0\\
-1 & 1 & 0\\
0 & -1 & 1
\end{pmatrix*},\
N=
\begin{pmatrix*}[r]
1 & 0.8 & 0.9\\
0 & 0.3 & 0\\
0 & 0   & 0.3
\end{pmatrix*},\
N^{-1}=
\begin{pmatrix*}[r]
1 & -\frac{8}{3} & -3\\[0.3em]
0 & \frac{10}{3} & 0\\[0.3em]
0 & 0            & \frac{10}{3}
\end{pmatrix*}
$$
Let us compute the super-hedging portfolio at time $k=0$. The matrix $T_0$ is diagonal
with all entries equal to $0.01$ and we need to compute the numbers
$C_{\max}(F,1)(\rho_i)$, which are as follows
\begin{align*}
C_{\max}(F,1)(\rho_0) 
&= 0\\
C_{\max}(F,1)(\rho_1) 
&= \frac{31}{6}\\
C_{\max}(F,1)(\rho_2)&= \frac{1}{3}\left( 14.5 + 32.5 \right)=\frac{94}{6}.
\end{align*}
The product of matrices $T_0^{-1}N^{-1}Q$ is equal to
$$
\frac{1}{300}
\begin{pmatrix*}[r]
11  & 1   & -9\\
-10 & 10  & 0\\
0   & -10 & 10
\end{pmatrix*}
$$
and computing $T_0^{-1}N^{-1}Q\cdot c_*$ gives the super-hedging 
portfolio 
$$
\alpha(0)=\left( -\frac{815}{1800},\frac{310}{1800},\frac{630}{1800} \right)
$$ 
at time $k=0$.
Observe that $q_j=\frac{1}{3}$ for all $j=0,1,2$ and hence $q_J=\frac{1}{9}$ for
all $J\in \C P_2(2)$. We thus get
$$
C_{\max}(F,0) = \frac{125}{18} 
$$
and the setup cost of the super-hedging portfolio is $V_{\alpha}(0) = \frac{125}{18}$ as claimed in Theorem \ref{T:minimum cost hedgiging theorem}.
Let us verify that it is a super-hedge. Using the notation \eqref{eqn:V_beta+1} of Section \ref{S:proof}, we get
\begin{align*}
V_{\alpha}^{+1}(0)
\left(  
\begin{smallmatrix}
0\\
0
\end{smallmatrix}
\right)
&=\frac{1}{1800}\left( -815\cdot S_0(1) + 310\cdot S_1(1)D_1 + 630\cdot S_2(1)D_2 \right)\\
&=0 = C_{\max}(F,1)(\rho_0),\\
V_{\alpha}^{+1}(0)
\left(  
\begin{smallmatrix}
1\\
0
\end{smallmatrix}
\right)
&=\frac{1}{1800}\left( -815\cdot S_0(1) + 310\cdot S_1(1)U_1 + 630\cdot S_2(1)D_2 \right)\\
&=\frac{31}{6} = C_{\max}(F,1)(\rho_1),\\
V_{\alpha}^{+1}(0)
\left(  
\begin{smallmatrix}
0\\
1
\end{smallmatrix}
\right)
&=\frac{1}{1800}\left( -815\cdot S_0(1) + 310\cdot S_1(1)D_1 + 630\cdot S_2(1)U_2 \right)\\
&=10.5 \geq \frac{16}{3} = 
C_{\max}(F,1)
\left( 
\begin{smallmatrix}
0\\
1
\end{smallmatrix}
 \right),
\\
V_{\alpha}^{+1}(0)
\left(  
\begin{smallmatrix}
1\\
1
\end{smallmatrix}
\right)
&=\frac{1}{1800}\left( -815\cdot S_0(1) + 310\cdot S_1(1)U_1 + 630\cdot S_2(1)U_2 \right)\\
&=\frac{94}{6} = C_{\max}(F,1)(\rho_2).
\end{align*}
Also, as claimed in Theorem \ref{T:minimum cost hedgiging theorem}, the strategy is a super-hedge and its value for $\rho_j$ is equal to $C_{\max}(F,1)(\rho_j)$.
It follows that the local residual at $\left(\begin{smallmatrix}
	0\\
	1
\end{smallmatrix}\right)$ is given by
$$
\delta_{\alpha}(1)
\left(  
\begin{smallmatrix}
0\\
1
\end{smallmatrix}
\right)
=10.5 - \frac{16}{3} \approx 5.17
$$
and is zero at all other states of the world.

Let, for example, $\omega = \left( 
\begin{smallmatrix}
0\\
1
\end{smallmatrix}
\right)
$ be a fixed state of the world at time $k=1$ and let us analyze the behavior of
the model from that state on.
It means that $S_1(1)(\omega) = 80$
and $S_2(1)(\omega) = 120$. We have that
$$
T_{1}(\omega)=
\begin{pmatrix}
100 & 0  & 0\\
0   & 80 & 0\\
0   & 0  & 120
\end{pmatrix}
$$
and we need to compute $C_{\max}(F,2)(\omega\rho_i)$ for $i=0,1,2$. These are the
following values: 
\begin{align*}
C_{\max}(F,2)(\omega\rho_0) 
&= F\left(  
\begin{smallmatrix}
0 & 0\\
1 & 0
\end{smallmatrix}
\right)
= \left( \frac{1}{2}S_1(0)D_1^2 + \frac{1}{2}S_2(0)U_2D_2 - 100\right)=0,
\\
C_{\max}(F,2)(\omega\rho_1) 
&= F\left(  
\begin{smallmatrix}
0 & 1\\
1 & 0
\end{smallmatrix}
\right)
= \left( \frac{1}{2}S_1(0)D_1U_1 + \frac{1}{2}S_2(0)U_2D_2 - 100\right)=0,
\\
C_{\max}(F,2)(\omega\rho_2) 
&= F\left(  
\begin{smallmatrix}
0 & 1\\
1 & 1
\end{smallmatrix}
\right)
= \left( \frac{1}{2}S_1(0)D_1U_2 + \frac{1}{2}S_2(0)U_2^2 - 100\right)=16.
\end{align*}
The matrix $T_1(\omega)^{-1}N^{-1}Q$ is given by
$$
\frac{1}{1800}
\begin{pmatrix*}[r]
66  & 6   & -54\\
-75 & 75  & 0\\
0   & -50 & 50
\end{pmatrix*}
$$
and hence the hedging portfolio at time $k=1$ under scenario $\omega$ is
$$
\alpha(1)(\omega) = \left( -\frac{108}{225},0,\frac{4}{9} \right).
$$
Observe that its value is equal to $C_{\max}(F,1)(\omega)$, according to formula \eqref{Eq:Cmax1}. The local residual at $\left(  
\begin{smallmatrix}
	0 & 0\\
	1 & 1
\end{smallmatrix}
\right)$ is given by
$$
\delta_{\alpha}(2)
\left(  
\begin{smallmatrix}
0 & 0\\
1 & 1
\end{smallmatrix}
\right)
=\Phi_{\alpha}(2)
\left(  
\begin{smallmatrix}
0 & 0\\
1 & 1
\end{smallmatrix}
\right)
-
F
\left(  
\begin{smallmatrix}
0 & 0\\
1 & 1
\end{smallmatrix}
\right)
=
16 - 4 = 12
$$
and it is zero at all other states of the world; that is, for $
\left( 
\begin{smallmatrix}
0 & *\\
1 & *
\end{smallmatrix}
 \right)
$, where the second column is different from $\omega = 
\left(  
\begin{smallmatrix}
0\\
1
\end{smallmatrix}
\right)
$.

\section{Conclusion and future research}\label{S:conclusion} 
This paper continues the studies of a multi-step discrete time model of the financial market, where $m\geq 2$ risky assets each follow a binomial model, and no assumptions are made on the joint distribution of the risky asset price processes. This model is incomplete, in contrast with its classical counterpart, the Cox-Ross-Rubinstein binomial model for $m=1$. 

In the paper  \cite{MR4553406}, we derived explicit formulas for the bounds of the no-arbitrage price interval for a wide class of multi-asset contingent claims, including European basket calls and puts.

In this paper, we complement the results of \cite{MR4553406} by providing explicit formulas for the minimum cost super-hedging strategy for a contingent claim of the same class. Such a super-hedge is a non-self-financing arbitrage strategy. At each time step, it produces a non-negative local residual. We provide explicit formulas for the local residuals as well. 

The multi-asset binomial model considered in this paper is not a realistic market model. However, our results (especially the explicit formulas obtained) create a foundation for further extensions of this model to more realistic market models. In our following paper \cite{2301.04996}, we generalize our results to the case of an extended multi-asset binomial model, where risky asset price ratios are continuously distributed over bounded intervals. In the future, we plan on building efficient algorithms of optimal hedging in this extended model based on our explicit formulas. See \cite{JKS-risk,MR3612260} and references therein for similar research on optimal hedging in an extended single-asset binomial model.

\section*{}

\bibliography{bibliography}
\bibliographystyle{plain}

\end{document}